\documentclass[xcolor=svgnames,11pt,reqno]{amsart}
\usepackage[round]{natbib}
\usepackage[T1]{fontenc}
\usepackage{amsmath}
\usepackage{amssymb}
\usepackage{amsthm}
\usepackage{fancyhdr}
\usepackage{enumitem}
\usepackage{comment}
\usepackage{relsize}
\usepackage{nicefrac}
\usepackage{bm}
\usepackage{mathrsfs}
\usepackage{graphicx}
\usepackage[utf8]{inputenc}
\usepackage{cancel}
\usepackage{mathtools}
\usepackage{tikz}
\usetikzlibrary{cd,decorations.pathreplacing,matrix,arrows,positioning,automata,shapes,shadows,calc,fadings,decorations,through,intersections}
\usepackage{float}
\usepackage{subcaption}
\usepackage{multirow}
\usepackage[left=2.5cm,top=2.8cm,right=2.5cm,bottom=2.5cm,head=3cm,headsep=1cm, foot=3cm]{geometry}
\usepackage{hyperref}
\usepackage[skip=2mm,indent=12pt]{parskip}
\definecolor{darkblue}{rgb}{0.1,0.1,0.9}
\definecolor{darkred}{rgb}{0.9,0.1,0.1}

\hypersetup{colorlinks, allcolors= darkblue}

\newtheorem{theorem}{Theorem}[section]
\newtheorem{definition}[theorem]{Definition}
\newtheorem{proposition}[theorem]{Proposition}
\newtheorem{assumption}[theorem]{Assumption}
\newtheorem{corollary}[theorem]{Corollary}
\newtheorem{lemma}[theorem]{Lemma}
\newtheorem{remark}[theorem]{Remark}
\renewcommand{\Pr}{\mathbb{P}}
\newcommand{\R}{\mathbb{R}}
\newcommand{\Q}{\mathbb{Q}}
\newcommand{\E}{\mathbb{E}}
\def\ccv{\preccurlyeq_{\hspace{- 0.4 mm} _{CCV}}}
\def\stccv{\prec_{\hspace{- 0.4 mm} _{CCV}}}
\def\icv{\preccurlyeq_{\hspace{- 0.4 mm} _{ICV}}}
\def\sticv{\prec_{\hspace{- 0.4 mm} _{ICV}}}
\renewcommand{\tilde}{\widetilde}
\newcommand{\q}{\quad}
\newcommand{\One}{\mathbf{1}}

\allowdisplaybreaks
\makeatletter

\newcommand{\Rmnum}[1]{\expandafter\@slowromancap\romannumeral #1@}

\title[Efficiency with Risk-Averse Monetary Utilities]{\bf Efficiency in Pure-Exchange Economies\vspace{0.25cm}\\with Risk-Averse Monetary Utilities\vspace{0.3cm}}

\author[Mario Ghossoub and Michael B.\ Zhu]{Mario Ghossoub\vspace{0.1cm}\\ University of Waterloo\vspace{0.8cm}\\ Michael B.\ Zhu\vspace{0.1cm}\\University of Waterloo\vspace{0.8cm}\\\today}

\address{{\bf Mario Ghossoub}: University of Waterloo -- Department of Statistics and Actuarial Science -- 200 University Ave.\ W.\ -- Waterloo, ON, N2L 3G1 -- Canada}
\email{\href{mailto:mario.ghossoub@uwaterloo.ca}{mario.ghossoub@uwaterloo.ca}\vspace{0.2cm}}

\address{{\bf Michael B.\ Zhu}:  University of Waterloo -- Department of Statistics and Actuarial Science -- 200 University Ave.\ W.\ -- Waterloo, ON, N2L 3G1 -- Canada}
\email{\href{mailto:mbzhu@uwaterloo.ca}{mbzhu@uwaterloo.ca}\vspace{0.4cm}}

\thanks{\vspace{0.2cm}}

\thanks{We are grateful to Patrick Beissner, Felix Liebrich, Frank Riedel, and Ruodu Wang and  for comments and suggestions. We thank audiences at the 2024 workshop on the Foundations and Applications of Decentralized Risk Sharing (FADeRiS 2024) at the University of Ulm and the 2024 workshop on Optimal Transport and Distributional Robustness at the Banff International Research Station. Mario Ghossoub acknowledges financial support from the Natural Sciences and Engineering Research Council of Canada (NSERC Grant No.\ 2018-03961 and 2024-03744). Michael B.\ Zhu acknowledges financial support from the Society of Actuaries through the Hickman Scholars Program.}

\begin{document}
\sloppy

\maketitle

\begin{abstract}
We study Pareto efficiency in a pure-exchange economy where agents' preferences are represented by risk-averse monetary utilities. These coincide with law-invariant monetary utilities, and they can be shown to correspond to the class of monotone, (quasi-)concave, Schur concave, and translation-invariant utility functionals. This covers a large class of utility functionals, including a variety of law-invariant robust utilities. We show that Pareto optima exist and are comonotone, and we provide a crisp characterization thereof in the case of law-invariant positively homogeneous monetary utilities. This characterization provides an easily implementable algorithm that fully determines the shape of Pareto-optimal allocations. Additionally, for positively homogeneous law-invariant monetary utilities, we show the existence of competitive equilibria and establish the first and second welfare theorems. In the special case of law-invariant comonotone-additive monetary utility functionals (concave Yaari-Dual utilities), we provide a closed-form characterization of Pareto optima. As an application, we examine risk-sharing markets where all agents evaluate risk through law-invariant coherent risk measures, a widely popular class of risk measures. In a numerical illustration, we characterize Pareto-optimal risk-sharing for some special types of coherent risk measures.
\end{abstract}

\bigskip
\section{Introduction}
\label{Intro}

At their core, risk sharing markets exist because individuals have different levels of risk aversion, whence Pareto-improving exchange ensues. In a context of choice under objective uncertainty, risk aversion is defined as consistency with second-order stochastic dominance (SSD), as in \cite{rothschild1978increasing}, for instance. Preferences that are SSD preserving are known as strongly risk averse preferences, and their study in the literature is vast. The study of optimal allocations in pure-exchange economies with risk-averse agents has its roots in the seminal work of \cite{borch1962} and \cite{wilson1968theory} in the framework of Expected Utility Theory (EUT), who show that each individual's optimal allocation must be a non-decreasing deterministic function of the aggregate endowment in the market. In particular, optimal allocations are \emph{comonotone}. This important property of allocations has been shown to extend to more general models of preferences. A landmark result in this direction is the \textit{comonotone improvement theorem} of \cite{landsberger1994co}, who show that for any given allocation of an aggregate risk between two agents, a Pareto-improving comonotone feasible allocation is always possible. The only requirement is that the agents' preferences must exhibit strong risk aversion, or more generally, be \textit{Schur concave}, i.e., consistent with the concave stochastic order. While the original result was established for a discrete state space, the comonotone improvement result was extended by \cite{dana2003modelling} to atomless probability spaces, by \cite{ludkovski2008comonotonicity} to more than two agents and unbounded random variables, and by \cite{CarlierDanaGalichon2012} to vectors of random variables. Recently, \cite{denuitetal2023comonotonicity} provide a proof of the comonotone improvement result in the general case, using a constructive algorithmic approach. Comonotonicity of allocations is important enough to warrant the study of the so-called \textit{comonotone market}, an incomplete market in which only comonotone allocations are available (see, e.g., \cite{boonen2021competitive}).

\medskip

The comonotone improvement theorem laid the groundwork for an extensive literature focusing on the characterization of Pareto-optimal allocations under different models of agents' preferences. The original result of \cite{borch1962} includes an explicit formula for optimal allocations as a function of each agent's decreasing marginal utility of wealth (the so-called Borch rule). Allocations in markets with Knightian uncertainty, as modeled through either the Choquet Expected Utility model of \cite{schmeidler1989subjective} or the Maxmin Expected Utility model of \cite{gilboa1989maxmin}, are studied by \cite{chateauneuf2000optimal}, \cite{dana2004ambiguity}, \cite{tsanakas2006risk}, \cite{decastro2011ambiguity}, and \cite{beissner2023optimal}, for instance. Through the relationship of these models with the classical EUT framework, it is possible to characterize the shape of Pareto-optimal allocations. However, for more general models of preferences such as the popular class of monetary utilities (i.e., monotone, concave, and translation-invariant utility functionals), explicit characterizations of Pareto-optimal allocations are more difficult to obtain. Instead, the focus has been to show the existence of Pareto-optimal allocations, within the set of comonotone allocations. A seminal result in this direction is due to \cite{JouiniSchachermayerTouzi2008}, who prove existence of comonotone allocations between two agents whose preferences are represented by law-invariant monetary utilities. In particular, given translation invariance, identifying Pareto optima reduces to solving the \emph{sup-convolution} of preferences, as observed by \cite{Barrieu2005}, for instance. \cite{filipovic2008optimal} extend this result to markets with more than two agents, as well as to more general spaces of random variables that allow for unbounded endowments. Further results concerning existence of Pareto optima include \cite{Acciaio2007} for non-monotone preferences, \cite{MastrogiacomoRosazza} for quasi-concave utilities, and \cite{RavanelliSvindland2014} for a class of law-invariant variational preferences. Recently, it has been shown that for SSD-preserving and translation invariant preferences, the assumption of concavity can be relaxed (e.g., \cite{mao2020risk} and \cite{liebrich2021risk}). Pareto-optimal allocations have also be proven to exist in more complicated market models, including markets with a portfolio of multiple assets (e.g., \cite{KieselRuschendorf2010} and \cite{KieselRuschendorf2014}) and markets with predetermined sets of admissible endowments (e.g., \cite{LiebrichSvindland2019}). However, results concerning the shape of Pareto-optimal allocations are more rare in this strand of the literature, and require more stringent assumptions on agents' preferences. Notably, \cite{embrechts2018quantile} and \cite{liu2022inf} provide closed-form characterizations of Pareto-optimal allocations when preferences are given by quantile-based risk measures such as the expected shortfall; and \cite{liu2020weighted} extends these results to distortion risk measures.

\medskip

In this paper, our main result (Theorem \ref{thm:cpo_coherent}) provides a characterization of Pareto optima for law-invariant and positively homogeneous monetary utilities (hence SSD preserving). Our characterization relies on the dual representation of these functionals in the spirit of \cite{Kusuoka2001}, who shows that these functionals can be expressed as an infimum of expectations over a certain set of probability measures. This representation is generalized to concave Schur-concave functionals by \cite{dana2005representation}, who also provides a representation in terms of the dual utilities of \cite{yaari}. The key to our result is expressing the sup-convolution problem in terms of the dual representation of preferences. The supremum from the sup-convolution and the infimum from the dual representation can be exchanged, and the resulting expression simplifies to yield our characterization. Notably, this transforms the domain of the optimization problem from the space of feasible allocations to the space of distortion functions. Once an appropriate problem is solved, the structure of the (comonotone) Pareto-optimal allocations is determined explicitly.

\medskip

In the context of risk measures, the dual concept to a positively homogeneous monetary utility is the notion of a coherent risk measure introduced by \cite{artzner1999coherent}. Indeed, these notions are equivalent up to a change in sign: if $U$ is a positively homogeneous monetary utility, then $\rho := -U$ is a coherent risk measure. Our result therefore also provides a characterization of Pareto optima in the case of a risk-sharing problem among multiple agents who evaluate risk via law-invariant coherent risk measures. As a special case, we show that when each risk measure is comonotone-additive, we recover the closed-form characterization obtained by \cite{liu2020weighted}. In a numerical illustration, we examine Pareto-optimal risk-sharing allocations in a market where a regulatory entity imposes capital requirements to all agents. In such a market, we apply an algorithm based on our main result to characterize the unique comonotone Pareto-optimal allocation. 

\medskip

The remainder of this paper is structured as follows. Section \ref{sec:formulation} introduces the optimal allocation problem. Section \ref{sec:comonotone} provides some background on the comonotone improvement result for Schur-concave functionals, and motivates the notion of a comonotone market. Our main characterization result for law-invariant and positively homogeneous monetary utilities is provided in Section \ref{sec:explicit}. Section \ref{sec:risk_sharing} provides an application to risk-sharing markets with coherent risk measures, and includes some numerical illustrations. Section \ref{sec:conclusion} concludes. The proofs of our main results, as well as some related analysis, can be found in the \hyperlink{LinkToAppendix}{Appendix}.

\bigskip
\section{Problem Formulation}
\label{sec:formulation}

Let $\mathcal{X}:=L^\infty\left(\Omega,\mathcal{F},\mathbb{P}\right)$ be the set of essentially bounded random variables on an atomless probability space $\left(\Omega,\mathcal{F},\mathbb{P}\right)$. There are $n\in\mathbb{N}$ agents wishing to reallocate their initial endowments among themselves without central authority involvement. For each $i \in \mathcal{N} :=\{1,\ldots,n\}$, let $X_i \in \mathcal{X}$ denote the initial endowment of the $i$-th agent. We consider a one-period economy, where all financial gains and losses are realized at the end of the period. Denote the aggregate endowment by $S:=\sum_{i=1}^nX_i$. Since $S$ is a sum of essentially bounded random variables, it is also essentially bounded by some constant $M<\infty$.

At the end of the period, the aggregate endowment $S$ is redistributed among the agents in the market. For each $i \in \mathcal{N}$, we denote the end-of-period, post-transfer payout of agent $i$ by $Y_i$. Therefore, the \textit{ex ante} admissible set of decision variables, henceforth referred to as the set of allocations, is given by
\begin{equation*}
\mathcal{A}:=\left\{\left\{Y_i\right\}_{i=1}^{n}\in\mathcal{X}^n:\sum_{i=1}^{n}Y_i=S\right\}.
\end{equation*}

We assume that each agent $i \in \mathcal{N}$ has a preference $\succ_i$ on $\mathcal{X}$ that admits a representation by a utility functional $U_i:\mathcal{X}\to\R$. We recall below some standard properties of such functionals.

\begin{definition}
    A utility functional $U:\mathcal{X}\to\R$ is said to be:

    \begin{itemize}
        \item \emph{Monotone} if $U(Z_1)\le U(Z_2)$, for all $Z_1,Z_2\in\mathcal{X}$ such that $Z_1\le Z_2$.
        \item \emph{Translation invariant} if $U(Z+c)=U(Z)+c$, for all $Z\in\mathcal{X}$ and $c\in\R$.
        \item \emph{Concave} if for all $Z_1,Z_2\in\mathcal{X}$ and $t\in[0,1]$,
            \[
                t \, U(Z_1) + (1-t) \, U(Z_2)\le U(t \, Z_1 + (1-t) \, Z_2)\,.
            \]
        \item \emph{Positively Homogeneous} if $U(t \, Z) = t \, U(Z)$, for all $Z\in\mathcal{X}$ and $t\ge0$.
        \item \emph{Law-invariant} if for all $Z_1,Z_2\in\mathcal{X}$ with the same distribution under $\Pr$, we have $U(Z_1)=U(Z_2)$.
    \end{itemize}
\end{definition}

\smallskip

\begin{definition}
\label{defn:ir_po}
    An allocation $\left\{Y^*_i\right\}_{i=1}^{n}\in\mathcal{A}$ is said to be:

    \begin{itemize}
        \item {\bf Individually Rational} (IR) if it incentivizes the parties to participate in the market, that is, 
            \begin{equation*}
                U_i\left(Y^*_i\right)\geq U_i\left(X_i\right),\quad \forall \, i \in \mathcal{N}.
            \end{equation*}
        
        \item {\bf Pareto Optimal} (PO) if it is IR and there does not exist any other IR allocation $\left\{Y_i\right\}_{i=1}^{n}$ such that
        \begin{equation*}
            U_i\left(Y_i\right)\geq U_i\left(Y^*_i\right),\quad \forall \, i \in \mathcal{N},
        \end{equation*}
        with at least one strict inequality.

        \item {\bf Weakly Pareto Optimal} if it is IR and there does not exist any other IR allocation $\left\{Y_i\right\}_{i=1}^{n}$ such that
    \begin{equation*}
        U_i\left(Y_i\right)>U_i\left(Y^*_i\right),\quad \forall \, i \in \mathcal{N}.
    \end{equation*}
    \end{itemize}
\end{definition}

It follows immediately that if an allocation is Pareto optimal, then it is weakly Pareto optimal. The converse holds under fairly mild conditions on the monotonicity of the preferences. For example, in \cite{xia2004multi}, these notions are equivalent in the expected-utility setting when utility functions are strictly increasing. For our purposes, the following condition will suffice. A similar condition can be found in \cite{RavanelliSvindland2014}, where it is used for the same purpose.

\begin{assumption}
    \label{as:increasing_wrt_premium}
    For all $i\in\mathcal{N}$ and $Z\in\mathcal{X}$, the function 
    \begin{align*}
    \mathbb{R} &\to \mathbb{R}\\
    c & \mapsto U_i(Z+c)
    \end{align*}
    is continuous, strictly increasing, and satisfies $\underset{c\to\infty}\lim\,U_i(Z+c)=\infty$.
\end{assumption}

\smallskip

\begin{lemma}
    \label{lem:po_strict}
    Under Assumption \ref{as:increasing_wrt_premium}, an allocation $\{Y_i^*\}_{i=1}^n\in\mathcal{A}$ is PO if and only if it is weakly Pareto-optimal.
\end{lemma}
\begin{proof}
    Clearly, Pareto optimality implies weak Pareto optimality. For the reverse direction, suppose that $\{Y_i^*\}_{i=1}^n\in\mathcal{A}$ is not PO. Then there exists an IR allocation $\left\{Y_i\right\}_{i=1}^{n}$ such that
    \begin{equation*}
        U_i\left(Y_i\right)\ge U_i\left(Y^*_i\right),\quad \forall \, i \in \mathcal{N}\,,
    \end{equation*}
    with at least one strict inequality. Without loss of generality, suppose that the inequality is strict for $U_1$. Then by Assumption \ref{as:increasing_wrt_premium}, there exists $\varepsilon>0$ such that
    \[
        U_1\left(Y_1-\varepsilon\right)>U_1\left(Y_1^*\right)\,.
    \]
    Let $\tilde{Y}_1:=Y_1-\varepsilon$ and $\tilde{Y}_i:=Y_i+\left(\frac{1}{n-1}\right)\varepsilon$, for $i=2,\ldots,n$. Then $\sum_{i=1}^n\tilde{Y}_i=\sum_{i=1}^nY_i=S$, and hence $\left\{\tilde{Y}_i\right\}_{i=1}^{n}\in\mathcal{A}$. Furthermore, for $i=2,\ldots,n$, we have
    \[
        U_i\left(\tilde{Y}_i\right)>U_i\left(Y_i\right)\ge U_i\left(Y_i^*\right),
    \]
    implying that $\{Y_i^*\}_{i=1}^n\in\mathcal{A}$ is not weakly Pareto-optimal either.
\end{proof}

Let $\mathcal{IR}$ denote the set of all IR allocations. Then, in particular, $\mathcal{IR} \neq \varnothing$ since it contains the status-quo, i.e., the no-risk-sharing allocation $Y_i=X_i$ under which each agent retains their initial endowment.

Let $\mathcal{PO} \subseteq \mathcal{IR}$ denote the set of all PO allocations. The following provides a useful characterization of the set $\mathcal{PO}$ under the assumption that all $U_i$ are concave. Define the set $\Lambda$ as follows:
\[\Lambda:=\{(\lambda_1,\ldots,\lambda_n)\in\R_+^n\}\setminus\{0\}\,.\]

\noindent For a given vector $\lambda\in\Lambda$, let $\mathcal{S}_\lambda$ be the set of all maximizers for the following sum-maximization problem:
\begin{equation}
    \label{eq:weighted_infconv}
    \sup_{\{Y_i\}_{i=1}^n\in\mathcal{IR}} \ \sum_{i=1}^{n}\lambda_i\,U_i(Y_i)\,.
\end{equation}

\smallskip

\noindent By a classical result, the set $\mathcal{PO}$ coincides with the solutions to problem \eqref{eq:weighted_infconv} for some choice of $\lambda\in\Lambda$. Elements of $\Lambda$ in this context are often referred to as Negishi weights in the literature. Hence, every Pareto-optimal allocation is associated with a vector of Negishi weights by this characterization. 

For completeness, we provide a proof of the following result in the Appendix. The argument of the proof follows the framework of \cite[Proposition 6.3.3]{dana2003financial}, who show this result without the consideration of individual rationality. The included proof shows that the individual rationality constraint is compatible with the classical methodology.

\begin{proposition}
\label{prop:weighted_infconv}
Suppose that Assumption \ref{as:increasing_wrt_premium} holds and that $U_i$ is monotone and concave for all $i\in\mathcal{N}$. Then $\mathcal{PO}=\underset{\lambda\in\Lambda}\bigcup\mathcal{S}_\lambda$.
\end{proposition}

\bigskip
\section{Comonotone Allocations}
\label{sec:comonotone}

While more explicit characterizations of solutions to \eqref{eq:weighted_infconv} are difficult to determine in practice, it is well known that Pareto-optimal allocations are comonotone\footnote{A random vector $\{Z_i\}_{i=1}^n$ is said to be \emph{comonotone} if $\left[Z_i(\omega_1)-Z_j(\omega_2)\right]\left[Z_i(\omega_1)-Z_j(\omega_2)\right]\ge0$, for all $\omega_1,\omega_2\in \Omega$ and $i,j\in\{1,\dots,n\}$.}
when each agent's preference preserves the concave order (i.e., preferences are Schur concave). This classical result was first obtained by  \cite{landsberger1994co} in the two-agent case over a discrete state space, and later extended to continuous random variables by \cite{dana2003modelling} and  \cite{filipovic2008optimal}.

The comonotone improvement theorem motivates the notion of a comonotone market. The latter is a special incomplete market where the only admissible allocations are those that are comonotone with the aggregate endowment. This market was previously studied by  \cite{boonen2021competitive}, who provide characterizations of competitive equilibria when each agent's preference admits a representation in terms of the Rank-Dependent Expected-Utility (RDU) model of \cite{quiggin82}. In this section, we examine the relationship between Pareto optima in a comonotone market and those in the general case. We find that when preferences are SSD preserving, each Pareto-optimal allocation in the comonotone market is also Pareto-optimal in the original unconstrained market, therefore providing justification for the existence of the comonotone market itself.

\medskip
\subsection{Schur Concave and SSD-Preserving Maps}
\label{subsec:convex}

First, we recall some standard definitions for the concave and increasing concave orders. For an in-depth overview of the mathematical properties of these orders, we refer to \cite[Section 3.A]{shaked2007stochastic}.

\begin{definition}[Concave Order]
\label{defn:concave_order}
    For all random variables $Z_1,Z_2\in\mathcal{X}$, we say that $Z_2$ dominates $Z_1$ in the concave order, and we write
        \[
            Z_1\ccv Z_2\,,
        \]
    if and only if for every concave function $\phi$,
        \[
            \E[\phi(Z_1)]\le\E[\phi(Z_2)]\,,
        \]
    when the above expectations are defined. If, in addition, the above inequality is strict for every strictly concave function $\phi$, then $Z_2$ is said to strictly dominate $Z_1$ in the concave order, denoted by $Z_1 \stccv Z_2$.
\end{definition}

\smallskip

\begin{definition}[Schur Concavity]
    A functional $U:\mathcal{X}\to\R$ is Schur concave if for all $Z_1,Z_2\in\mathcal{X}$ such that $Z_1\ccv Z_2$, we have
        \[
            U(Z_1)\le U(Z_2)\,.
        \]
    Similarly, $U$ is strictly Schur concave if for all $Z_1,Z_2\in\mathcal{X}$ such that $Z_1\stccv Z_2$, we have
        \[
            U(Z_1)<U(Z_2)\,.
        \]
\end{definition}

\smallskip

\begin{definition}[Increasing Concave Order]
\label{defn:inc_concave_order}
    For all random variables $Z_1,Z_2$ on a probability space $(\Omega,\mathcal{F},\Pr)$, we say that $Z_2$ dominates $Z_1$ in the increasing concave order, and we write
    \[Z_1\icv Z_2\,,\]
    if and only if for every increasing concave function $\phi$,
    \[\E[\phi(Z_1)]\le\E[\phi(Z_2)]\,,\]
    when the above expectations are defined. If, in addition, the above inequality is strict for every increasing strictly concave function $\phi$, then $Z_2$ is said to strictly dominate $Z_1$ in the concave order, denoted by $Z_1\sticv Z_2$.
\end{definition}

In economic theory, the increasing concave order is a classical notion of risk aversion, and it is commonly referred to as second-order stochastic dominance. A preference functional that is consistent with respect to this order is said to be strongly risk averse. The literature on the behaviour-theoretic foundations of strong risk aversion is vast (e.g., \cite{rothschild1978increasing} and \cite{quiggin2012generalized}). We refer to these strongly risk averse maps as SSD-preserving maps, as is common in the literature. For more on preferences modelled directly by SSD-preserving maps, see \cite{dana2003modelling} for instance.

\begin{definition}[SSD Consistency]
    A functional $U:\mathcal{X}\to\R$ is SSD consistent (or SSD preserving) if for all $Z_1,Z_2\in\mathcal{X}$ such that $Z_1\icv Z_2$, we have
        \[
            U(Z_1)\le U(Z_2)\,.
        \]
    Similarly, $U$ is strictly SSD preserving if for all $Z_1,Z_2\in\mathcal{X}$ such that $Z_1\sticv Z_2$, we have
        \[
            U(Z_1)<U(Z_2)\,.
        \]
\end{definition}

For more on the relationship between the concave and increasing concave orders, we refer to \cite{shaked2007stochastic} and \cite{dana2005representation}. In particular, the following two results will be relevant for the present paper. The first result provides an equivalence between Schur concavity and SSD consistency for monotone functionals. The next result shows that given concavity and an additional regularity assumption, Schur concavity and SSD consistency are equivalent to law-invariance. Together, these results characterize monotone, concave, and law-invariant functionals as a class of strongly risk averse preferences.

\begin{proposition}
\label{prop:ssd_vs_schur}
    A functional $U:\mathcal{X}\to\R$ is SSD preserving if and only if it is monotone and Schur concave.
\end{proposition}
\begin{proof}
    See \cite[Proposition 2.1]{dana2005representation}.
\end{proof}

\begin{proposition}
\label{prop:schur_law_invariant}
    Suppose that a functional $U:\mathcal{X}\to\R$ is concave and upper semicontinuous with respect to the norm topology on $L^\infty$. Then Schur concavity of $U$ is equivalent to law-invariance of $U$.
\end{proposition}
\begin{proof}
    The result follows directly from the following three equivalences.    
    By \cite[Corollary 3.3]{grechuk2012schur}, for a concave Schur-concave functional, norm upper semicontinuity is equivalent to $\sigma(L^\infty,L^1)$ upper semicontinuity.
    By \cite[Theorem 4.1]{dana2005representation}, for a concave and $\sigma(L^\infty,L^1)$ upper semicontinuous functional, Schur concavity and law-invariance are equivalent.
    Finally, by \cite[Theorem 2.2]{JouiniSchachermayerTouzi2006}, for a concave law-invariant functional, $\sigma(L^\infty,L^1)$ upper semicontinuity and norm upper semicontinuity are equivalent.
\end{proof}

\smallskip

\begin{corollary}
    \label{cor:ssd_law_invariant}
    Suppose that a functional $U:\mathcal{X}\to\R$ is monotone, concave, and upper semicontinuous with respect to the norm topology on $L^\infty$. Then $U$ is SSD preserving if and only if it is law-invariant.  
\end{corollary}
\begin{proof}
    Direct corollary of Propositions \ref{prop:ssd_vs_schur} and \ref{prop:schur_law_invariant}.
\end{proof}

\medskip
\subsection{Pareto Optimality in Comonotone Markets}
\label{subsec:comonotone}

Denote the set of all comonotone allocations by
    \begin{equation*}
        \mathcal{A}_C:=\Big\{\left\{Y_i\right\}_{i=1}^{n}\in\mathcal{A}:
            \{Y_i\}_{i=1}^n\mbox{ is comonotone}\Big\}. 
    \end{equation*}

\smallskip

The notion of Pareto optimality in this comonotone market is defined below.

\begin{definition}
\label{defn:cpo}
    An allocation $\left\{Y^*_i\right\}_{i=1}^{n}\in\mathcal{A}_C$ is \emph{Comonotone Pareto Optimal} (CPO) if it is IR and there does not exist any other IR allocation $\left\{Y_i\right\}_{i=1}^{n}\in\mathcal{A}_C$ such that
        \begin{equation*}
            U_i\left(Y_i\right)\geq U_i\left(Y^*_i\right),
                \quad\forall\,i\in\mathcal{N},
        \end{equation*}
    with at least one strict inequality.
\end{definition}

That is, CPO allocations are those that are not dominated by any other \emph{comonotone} allocation, and are therefore Pareto optimal in the (restricted) comonotone market. However, we will show that given SSD consistency, CPO allocations are also Pareto optimal in the original market. Hence, there is no ambiguity in referring to these allocations as ``comonotone Pareto optimal''.

Let $\mathcal{CPO}$ denote the set of all CPO allocations. Analogous to the result of Proposition \ref{prop:weighted_infconv}, we can obtain the following characterization for the set $\mathcal{CPO}$. For a given $\lambda\in\Lambda$, let $\mathcal{CS}_\lambda$ be the set of all maximizers of the following sum-maximization problem:
    \begin{equation*}
        \sup_{\left\{Y_i\right\}_{i=1}^{n}\in\mathcal{IR}\cap\mathcal{A}_C}
            \ \sum_{i=1}^{n}\lambda_i\,U_i(Y_i).
    \end{equation*}

\begin{corollary}
\label{cor:comonotone_characterization}
     Suppose Assumption \ref{as:increasing_wrt_premium} holds. If $U_i$ is monotone and concave for all $i\in\mathcal{N}$, then $\mathcal{CPO}=\underset{\lambda\in\Lambda}\bigcup\mathcal{CS}_\lambda$.
\end{corollary}

\begin{proof}
    The proof is identical to that of Proposition \ref{prop:weighted_infconv}, with the set $\mathcal{IR}$ replaced by $\mathcal{IR}\cap\mathcal{A}_C$. Note that convex combinations of comonotone allocations are also comonotone.
\end{proof}

\smallskip

\begin{proposition}
\label{prop:comonotone_improvement}
    For each $\{Y_i\}_{i=1}^n\in\mathcal{A}$ there exists a comonotone allocation $\left\{\tilde{Y}_i\right\}_{i=1}^n\in\mathcal{A}_C$ such that
        \[
            Y_i\ccv\tilde{Y}_i,
                \quad\forall i\in\mathcal{N}\,.
        \]
    If, in addition, $\{Y_i\}_{i=1}^n\not\in\mathcal{A}_C$, then the comonotone allocation $\left\{\tilde{Y}_i\right\}_{i=1}^n\in\mathcal{A}_C$ can be taken such that
        \[
            Y_j\stccv\tilde{Y}_j\,,
        \]
    for some $j\in\mathcal{N}$.
\end{proposition}
\begin{proof}
    See, e.g., \cite[Theorem 3.1]{CarlierDanaGalichon2012} or \cite[Theorem 3.1]{denuitetal2023comonotonicity}.
\end{proof}

\smallskip

\begin{corollary}
\label{cor:comonotone_improvement_utilities}
    Suppose that $U_i$ is Schur concave, for each $i\in\mathcal{N}$. Then for each $\{Y_i\}_{i=1}^n\in\mathcal{A}$ there exists a comonotone allocation $\left\{\tilde{Y}_i\right\}_{i=1}^n\in\mathcal{A}_C$ such that
        \[
            U_i(Y_i)\le U_i(\tilde{Y}_i),
                \quad\forall i\in\mathcal{N}\,.
        \]

        \smallskip

        \noindent   If, in addition, $U_i$ is strictly Schur concave for each $i\in\mathcal{N}$ and $\{Y_i\}_{i=1}^n\not\in\mathcal{A}_C$, then there exists a comonotone allocation $\left\{\tilde{Y}_i\right\}_{i=1}^n\in\mathcal{A}_C$ such that
        \[
            U_j(Y_j)<U_j(\tilde{Y}_j)\,,
        \]
    for some $j\in\mathcal{N}$.
\end{corollary}
\begin{proof}
    Direct consequence of Proposition \ref{prop:comonotone_improvement} and Definition \ref{defn:concave_order}.
\end{proof}

The following results illustrate the relationship between PO allocations and CPO allocations, given Schur concavity of preferences.

\begin{theorem}
\label{thm:PO_vs_CPO}
    Suppose that each $U_i$ is concave and SSD preserving, and that Assumption \ref{as:increasing_wrt_premium} holds. Then $\mathcal{PO}\ne\varnothing$ if and only if $\mathcal{CPO}\ne\varnothing$. Furthermore, when Pareto optima and comonotone Pareto optima exist, we have for all $\lambda\in\Lambda$,
        \begin{equation}
        \label{eq:inf_convs_equal}
            \sup_{\{Y_i\}_{i=1}^{n}\in\mathcal{IR}}
                \left\{\sum_{i=1}^{n}\lambda_i\,U_i(Y_i)\right\}
                =\sup_{\{Y_i\}_{i=1}^{n}\in\mathcal{IR}\cap\mathcal{A}_C}
                \left\{\sum_{i=1}^{n}\lambda_i\,U_i(Y_i)\right\}\,.
        \end{equation}
\end{theorem}

We obtain the following result as a corollary, which describes the relationship between PO and CPO in this case.

\begin{corollary}
\label{cor:po_ac_cpo}
    Suppose that each $U_i$ is concave and SSD preserving, and Assumption \ref{as:increasing_wrt_premium} holds. Then $\mathcal{PO}\cap\mathcal{A}_C=\mathcal{CPO}$. If, in addition, each $U_i$ is strictly SSD preserving, then $\mathcal{PO}=\mathcal{CPO}$.
\end{corollary}
\begin{proof}
    Suppose that $\{Y_i^*\}_{i=1}^n\in\mathcal{PO}\cap\mathcal{A}_C$. Since this allocation is Pareto optimal, it is not Pareto-dominated by any other allocation. That is, there exists no other allocation $\{\tilde{Y_i}\}_{i=1}^n\in\mathcal{IR}$ such that $U_i(\tilde{Y_i})\ge U_i(Y_i^*)$ for all $i\in\mathcal{N}$, with at least one strict inequality. Hence, there exists no such allocation in $\mathcal{IR}\cap\mathcal{A}_C$ either. Since $\{Y_i^*\}_{i=1}^n$ is a comonotone allocation, this implies that $\{Y_i^*\}_{i=1}^n\in\mathcal{CPO}$, and so $\mathcal{PO}\cap\mathcal{A}_C\subseteq\mathcal{CPO}$.

\smallskip

    For the reverse inclusion, suppose that $\{Y_i^*\}_{i=1}^n\in\mathcal{CPO}$. Then by Theorem \ref{thm:PO_vs_CPO}, there exists $\lambda\in\Lambda$ for which the allocation $\{Y_i^*\}_{i=1}^n$ achieves both infima in \eqref{eq:inf_convs_equal}. Hence, $\{Y_i^*\}_{i=1}^n\in\mathcal{S}_\lambda\subseteq\mathcal{PO}$ by Proposition \ref{prop:weighted_infconv}. Now suppose that each $U_i$ is strictly Schur concave and that $\{Y_i\}_{i=1}^n\in\mathcal{A}\setminus\mathcal{A}_C$. Then by Corollary \ref{cor:comonotone_improvement_utilities}, there exists a comonotone allocation $\{\tilde{Y}_i\}_{i=1}^n$ such that
        \[
            U_i(Y_i)\le U_i(\tilde{Y}_i)\q\forall i\in\mathcal{N}\,,   
        \]
    with at least one strict inequality. Therefore, $\{Y_i\}_{i=1}^n\not\in\mathcal{PO}$, which implies $\mathcal{PO}=\mathcal{PO}\cap\mathcal{A}_C=\mathcal{CPO}$ by the above.
\end{proof}

Define the \emph{utility possibility frontier} to be the set of all $n$-dimensional real vectors that represent the utility achievable by each agent in a Pareto-optimal allocation. That is, 
\[
    \mathcal{UPF}:=\left\{(u_1,\ldots,u_n)\in\R^n:u_i=U_i(Y_i^*), \,\forall i\in\mathcal{N},\mbox{ for some }\{Y_i^*\}_{i=1}^n\in\mathcal{PO}\right\}\,.
\]

\noindent A similar concept can be defined for the comonotone market. Let the \emph{comonotone utility possibility frontier} be the utility achievable by each agent in a comonotone Pareto-optimal allocation:
\[
    \mathcal{CUPF}:=\left\{(u_1,\ldots,u_n)\in\R^n:u_i=U_i(Y_i^*), \, \forall i\in\mathcal{N},\mbox{ for some }\{Y_i^*\}_{i=1}^n\in\mathcal{CPO}\right\}.
\]

\medskip

A consequence of Theorem \ref{thm:PO_vs_CPO} is that for concave Schur-concave utilities, both of these possibility frontiers coincide, as shown by the following.

\medskip

\begin{corollary}
\label{cor:upf}
    Suppose that each $U_i$ is concave and SSD preserving, and that Assumption \ref{as:increasing_wrt_premium} holds. Then $\mathcal{UPF}=\mathcal{CUPF}$.
\end{corollary}
\begin{proof}
    By Corollary \ref{cor:po_ac_cpo}, we have $\mathcal{CPO}\subseteq\mathcal{PO}$, and so $\mathcal{CUPF}\subseteq\mathcal{UPF}$. For the reverse inclusion, suppose that $\{Y_i^*\}_{i=1}^n\in\mathcal{PO}$. Then by Corollary \ref{cor:comonotone_improvement_utilities}, there exists a comonotone allocation $\left\{\tilde{Y}_i\right\}_{i=1}^n\in\mathcal{A}_C$ such that
        \[
            U_i\left(\tilde{Y}_i\right)\ge U_i(Y_i^*)
                \quad\forall i\in\{1,\ldots,n\}\,.
        \]
    However, since $\{Y_i^*\}_{i=1}^n$ is Pareto optimal, equality must hold for each $i$, implying that $\mathcal{UPF}\subseteq\mathcal{CUPF}$.
\end{proof}

These results provide additional justification for the existence of a comonotone market whenever preferences are SSD preserving. In general, the set of all Pareto-optimal allocations is difficult to characterize. However, whenever Pareto optima exist, comonotone Pareto optima must also exist by Theorem \ref{thm:PO_vs_CPO}. Furthermore, Corollary \ref{cor:po_ac_cpo} implies that these comonotone Pareto optima are indeed Pareto optimal in the unconstrained market. Hence, imposing the restriction that all allocations must be comonotone does not adversely affect the total welfare gain that is possible, as highlighted by Corollary \ref{cor:upf}.

Finally, it is possible to show that $\mathcal{PO}\ne\varnothing$ (equivalently, $\mathcal{CPO}\ne\varnothing)$ under some additional conditions. Notably, by \cite[Theorem 2.5]{filipovic2008optimal}, comonotone Pareto optima exist when preferences are concave, Schur concave, translation invariant, and $\sigma(L^\infty,L^1)$ upper semicontinuous.

More generally, an example of a class of preferences that are concave and SSD preserving is the class of probabilistically sophisticated variational preferences of \cite{maccheroni2006ambiguity}. These preferences admit the representation
    \[
U(X)=\inf_{\mathbb{Q}\in\mathcal{Q}}\left(\E^\mathbb{Q}[u(X)]+\alpha(\mathbb{Q})\right)\,,
    \]
where $u$ is a suitable utility function, $\mathcal{Q}$ is closed under densities with the same distribution, and $\alpha$ is a suitable law-invariant functional. We refer to \cite{RavanelliSvindland2014} for more details on this class of preferences, which includes many common functionals as special cases (e.g., expected utilities and coherent risk measures). It can also be shown under mild assumptions that Pareto-optimal allocations exist when each agent has a probabilistically sophisticated variational preference (e.g., \citealt[Theorem 4.1]{RavanelliSvindland2014}).

\bigskip
\section{Characterization of Pareto Optima for Positively Homogeneous Risk-Averse Monetary Utilities}
\label{sec:explicit}

In this section, we provide the main result of this paper: a characterization of all (comonotone) Pareto-optimal allocations when each agent's preference is represented by a positively homogeneous risk-averse monetary utility. 

\medskip
\subsection{Monetary Utilities}
First, we recall the definition of a monetary utility.
\begin{definition}[Monetary Utility]
\label{defn:monetary}
    A functional $U:\mathcal{X}\to\R$ is a \emph{monetary utility} if it is monotone, concave, and translation invariant.
\end{definition}

\begin{remark}
    It can be shown that under translation invariance, quasi-concavity and concavity are equivalent for monotone functionals \cite[Proposition 5]{delbaen2012monetary}. Hence, concavity in Definition \ref{defn:monetary} can be replaced by quasi-concavity.
\end{remark}

We refer to \cite{delbaen2012monetary} for an extensive overview of monetary utilities. The key property of monetary utilities is translation invariance, which allows for a simplified characterization of $\mathcal{PO}$ and $\mathcal{CPO}$ in terms of solutions of a maximization problem. We note that in contrast to the result of Proposition \ref{prop:weighted_infconv}, the following characterization does not require the functionals to be concave. Let $\One\in\R^n$ denote the vector with $1$ in every entry. When all functionals $U_i$ are translation invariant, we obtain the following result.

\begin{proposition}
    \label{prop:characterization}
    If $U_i$ is translation invariant for all $i\in\mathcal{N}$, then $\mathcal{PO}=\mathcal{S}_\One$.
\end{proposition}

This is a well-known result in the literature, and implies that when preferences are translation invariant, the Negishi weights may be taken to be equal. When each of these weights is $1$, finding the Pareto-optimal allocations is equivalent to solving the problem
    \begin{equation*}
        \sup_{\left\{Y_i\right\}_{i=1}^{n}\in\mathcal{IR}}
            \ \sum_{i=1}^{n}U_i(Y_i)\,,
    \end{equation*}
which is also known as the sup-convolution of the preferences $U_i$, for $i\in\mathcal{N}$. A proof of this result can be found in \cite[Proposition 1]{embrechts2018quantile} in the context of risk sharing, without the consideration of individual rationality. For completeness, we provide a proof of Proposition \ref{prop:characterization} in the Appendix. This characterization also applies to comonotone Pareto optima, as per the following result.

\begin{corollary}
\label{cor:comonotone_characterization_translation_invar}
    If $U_i$ is translation invariant for all $i\in\mathcal{N}$, then $\mathcal{CPO}=\mathcal{CS}_\One$.
\end{corollary}
\begin{proof}
    The proof is identical to that of Proposition \ref{prop:characterization}, with the set $\mathcal{IR}$ replaced by $\mathcal{IR}\cap\mathcal{A}_C$. 
\end{proof}

The property of translation invariance has some convenient implications on the regularity of preferences. Firstly, it is immediate that under translation invariance, Assumption \ref{as:increasing_wrt_premium} is automatically satisfied. Furthermore, monetary utilities are norm continuous, which implies the following result.

\begin{corollary}
\label{cor:monetary_schur_law_invariant}
    A monetary utility $U$ is SSD preserving if and only if it is law-invariant.
\end{corollary}
\begin{proof}
    Monetary utilities are 1-Lipschitz continuous with respect to the supremum norm on $L^\infty$ (e.g., \citealt[Lemma 4.3]{FollmerSchied2016}). The result then follows from Corollary \ref{cor:ssd_law_invariant}.
\end{proof}

Corollary \ref{cor:monetary_schur_law_invariant} implies that when each $U_i$ is a law-invariant monetary utility, the results of Subsections \ref{subsec:convex} and \ref{subsec:comonotone} apply to describe Pareto optimality through comonotone allocations.

\begin{remark}
\label{rmk:concave_and_trans_invar}
    In the event that each $U_i$ is a law-invariant monetary utility, Propositions \ref{prop:weighted_infconv} and \ref{prop:characterization} imply that
$$\mathcal{PO} = \underset{\lambda\in\Lambda}\bigcup\mathcal{S}_\lambda=\mathcal{S}_\One 
\ \ \hbox{and} \ \ 
\mathcal{CPO}=\underset{\lambda\in\Lambda}\bigcup\mathcal{CS}_\lambda=\mathcal{CS}_\One.$$

\noindent Indeed, it is possible to show that for all $\lambda\in\Lambda$, translation invariance implies that $\mathcal{S}_\lambda\subseteq\mathcal{S}_\One$. 
\end{remark}

\medskip
\subsection{Law-Invariant Positively Homogeneous Monetary Utilities}

In the following, we obtain a more explicit characterization of the set $\mathcal{CPO}$, under the additional assumption that each utility functional is positively homogeneous. Positively homogeneous monetary utilities are well-studied in the insurance and risk management literature, where they are commonly known as coherent risk measures, as introduced by \cite{artzner1999coherent}. We will use the fact that these functionals admit a representation in terms of Choquet integrals. Some preliminaries are provided below.

\begin{definition}
    A set function $\upsilon:\mathcal{F}\to\R$ is a \emph{capacity} if:
    \begin{itemize}
        \item $\upsilon(\varnothing)=0$ and $\upsilon(\Omega)<\infty$; and,
        \medskip
        \item If $A,B\in\mathcal{F}$ are such that $A\subseteq B$, then $\upsilon(A)\le\upsilon(B)$.
    \end{itemize}    
\end{definition}

\begin{definition}
    The \emph{Choquet integral} of $Z\in\mathcal{X}$ with respect to a capacity $\upsilon$ is defined as
    \[\int X\,d\upsilon:=\int_0^\infty\upsilon(X>t)\,dt+\int_{-\infty}^0\left[\upsilon(X>t)-\upsilon(\Omega)\right]\,dt\,.\]
\end{definition}

\begin{definition}
    A function $T:[0,1]\to[0,1]$ is a \emph{distortion function} if is it non-decreasing and satisfies $T(0)=0$ and $T(1)=1$.
\end{definition}

\begin{definition}
    When two random variables $Y,Z\in\mathcal{X}$ have the same distribution, we use the notation $Y\stackrel{d}{\sim}Z$. A subset $H\subseteq\mathcal{X}$ is \emph{law-invariant} if $Y\in\mathcal{H}$ and $Y\stackrel{d}{\sim}Z$ implies $Z\in H$.
\end{definition}

The following representation result is due to \cite{dana2005representation}.

\begin{lemma}
\label{lem:coherent_rep}
    Let $\mathcal{U}:\mathcal{X}\to\R$ be a law-invariant and positively homogeneous monetary utility. Then there exists a $\sigma(L^1,L^\infty)$ closed, convex, law-invariant set $\mathcal{H}\subseteq L_+^1$ such that
        \[
            \E[H]=1\mbox{ for all }H\in\mathcal{H}\,,
        \]
    and
        \[
            \mathcal{U}(Z)=\inf_{H\in\mathcal{H}}\E[HZ]=\inf_{T\in\{\phi_H:H\in\mathcal{H}\}}\int Z\,dT\circ\Pr\,,
        \]

        \smallskip
        
    \noindent where for every $H\in L_+^1$, the function $\phi_H:[0,1]\to[0,1]$ is defined by
            $$\phi_H(x) := \int_0^xF_{H,\Pr}^{-1}(t)\,dt, \ \forall x \in [0,1].$$    
\end{lemma}

\begin{proof}
    See \cite[Corollary 4.3]{dana2005representation}.
\end{proof}

\smallskip

\begin{remark}
    By \cite[Proposition 5.1]{bellini2021law}, the set $H$ in the representation of Lemma \ref{lem:coherent_rep} can be taken to be in $L_+^\infty$ instead of in $L_+^1$. However, the statement of Lemma \ref{lem:coherent_rep} will suffice for our purposes.
\end{remark}

In a comonotone market, every admissible post-transfer payout $Y_i$ is comonotone with the aggregate endowment $S$. Denote by $\mathcal{X}^\uparrow$ the set of all random variables in $\mathcal{X}$ that are comonotone with $S$. Under this setting, we may restrict the domain of each utility functional to $\mathcal{X}^\uparrow$, which admits the following representation. The proof of this result is provided in the Appendix.

\begin{lemma}
\label{lem:distortion_set_convex}
    Let $U:\mathcal{X}\to\R$ be a law-invariant and positively homogeneous monetary utility. Then for all $Z\in\mathcal{X}^\uparrow$, we have
        \[
            U(Z)=\inf_{T\in\mathcal{T}}\int Z\,dT\circ\Pr\,,
        \]
    where $\mathcal{T}$ is a convex set of convex distortion functions that is sequentially closed under pointwise convergence.
\end{lemma}

To facilitate our characterization result, we first show that all comonotone allocations are translations of a suitable non-negative function of the aggregate endowment. First, define the set $\mathcal{G}$ by the following:
    \begin{equation*}
        \mathcal{G}:=\left\{\left\{g_i\right\}_{i=1}^{n}\,\middle|
            \,g_i:\mathbb{R}_+\rightarrow\mathbb{R}_+\text{ non-decreasing},
            \text{ and }\sum_{i=1}^{n}g_i\left(\cdot\right)=\text{Id}\right\}.
    \end{equation*}
Let $\underline{s}:=\mathrm{ess}\inf S$. Since $S\in L^\infty$, we have $\underline{s}>-\infty$. The following lemma shows that every comonotone allocation can be represented in terms of functions in $\mathcal{G}$.

\begin{lemma}
\label{lem:translation}
    Let $\{Y_i\}_{i=1}^n\in\mathcal{A}_C$. Then there exist functions $\{g_i\}_{i=1}^n\in\mathcal{G}$ and constants $\{c_i\}_{i=1}^n\in\R^n$ such that
        \[
            Y_i=g_i(S-\underline{s})+c_i\quad\forall i\in\mathcal{N}\,,
        \]
    and $\sum_{i=1}^nc_i=\underline{s}$.
\end{lemma}

\smallskip

\begin{proof}
    Let $\{Y_i\}_{i=1}^n$ be a comonotone allocation. By a standard result \cite[Proposition 4.5]{denneberg1994non}, there exists increasing 1-Lipschitz functions $f_i:\R\mapsto\R$ such that $f_i(S)=Y_i$ and $\sum_{i=1}^nf_i(x)=x$ for all $x\in\R$. Let $c_i:=f_i(\underline{s})$, and define the function $g_i$ by
        \begin{align*}
            g_i:\R_+&\to\R_+\\
            x&\mapsto f_i(\underline{s}+x)-f_i(\underline{s})\,.
        \end{align*}
    Then $g_i$ is non-decreasing, and
        \[
\sum_{i=1}^ng_i(x)=\sum_{i=1}^nf_i(\underline{s}+x)-\sum_{i=1}^nf_i(\underline{s})
                =\underline{s}+x-\underline{s}
                =x\,,
        \]
    for all $x\ge0$, which implies that $\{g_i\}_{i=1}^n\in\mathcal{G}$. It follows from the construction of $g_i$ that
        \[
            f_i(x)=g_i(x-\underline{s})+c_i\,,
        \]
    from which we conclude that $Y_i=f_i(S)=g_i(S-\underline{s})+c_i$ for all $i\in\mathcal{N}$.
\end{proof}

\smallskip

Our main result is given below. The interpretation of this result is discussed in Subsection \ref{subsec:algorithm}. 

\begin{theorem}
\label{thm:cpo_coherent}
    Suppose that for each $i\in\mathcal{N}$, the utility functional $U_i$ is a law-invariant and positively homogeneous monetary utility. Let $\mathcal{T}_i$ denote the representing convex set of convex distortions for $U_i$ that is sequentially closed under pointwise convergence, in the sense of Lemma \ref{lem:distortion_set_convex}. Then the following hold:

    \smallskip
    
    \begin{enumerate}[label=(\roman*)]
        \item There exists a solution to the problem
            \begin{equation}
            \label{eq:inf_problem}
                \min_{\{T_i\}_{i=1}^n\in \, \prod_{i=1}^n\mathcal{T}_i}\,\int_0^\infty\max_{i\in\mathcal{N}}\,\{T_i(\Pr(S>\underline{s}+x))\}\,dx\,.
            \end{equation}

            \smallskip
            
        \item A necessary condition for an allocation $\{Y_i^*\}_{i=1}^n$ to be comonotone Pareto optimal is that
                \[
                    Y_i^*=g_i^*(S-\underline{s})+c_i^*\,,
                \]
                
            \noindent where $\{c_i^*\}_{i=1}^n\in\R^n$ is chosen such that $\sum_{i=1}^nc_i^*=\underline{s}$ and $\{g_i^*(S-\underline{s})+c_i^*\}_{i=1}^n\in\mathcal{IR}$, and $\{g_i^*\}_{i=1}^n\in\mathcal{G}$ can be written in terms of the integrals of suitable functions $h_i$.     
            Specifically, for each $i\in\mathcal{N}$, we can write $g_i^*(x)=\int_0^xh_i(z)\,dz$, where each $h_i:\R_+\to[0,1]$ is a function such that for each $(T_1^*,\ldots,T_n^*)$ that solves \eqref{eq:inf_problem} and for almost every $x\in\R_+$, we have
                \[
                    \sum_{i\in L_{x,T_1^*,\ldots,T_n^*}}h_i(x)=1\mbox{ and }\sum_{i\in\mathcal{N}\setminus L_{x,T_1^*,\ldots,T_n^*}}h_i(x)=0\,,
                \]
            where
                \begin{align*}
                    L_{x,T_1^*,\ldots,T_n^*}&:=\left\{i\in\mathcal{N}:
                        T_i^*(\Pr(S>\underline{s}+x))=\max_{j\in\mathcal{N}}\{T_j^*(\Pr(S>\underline{s}+x))\}\right\}\,.
                \end{align*}
    \end{enumerate}
\end{theorem}

\medskip

\begin{remark}
    In the statements of Lemma \ref{lem:translation} and Theorem \ref{thm:cpo_coherent}, we may replace $\underline{s}$ with any lower bound of $S$. For example, if the aggregate endowment $S$ is non-negative, we may take $\underline{s}=0$ for simplicity.
\end{remark}

\smallskip

\begin{corollary}
    Suppose that for each $i\in\mathcal{N}$, the utility functional $U_i$ is a strictly SSD-preserving, law-invariant, and positively homogeneous monetary utility. A necessary condition for an allocation $\{Y_i^*\}_{i=1}^n$ to be Pareto optimal is that
        \[
            Y_i^*=g_i^*(S-\underline{s})+c_i^*\,,
        \]
    where $\sum_{i=1}^nc_i^*=\underline{s}$ and $g_i^*$ are of the form given in Theorem \ref{thm:cpo_coherent}. That is, $g_i^*(x)=\int_0^xh_i(z)\,dz$, where $h_i:\R_+\to[0,1]$ such that for every $(T_1^*,\ldots,T_n^*)$ that solves \eqref{eq:inf_problem} and for almost every $x\in\R_+$,   
        \[
            \sum_{i\in L_{x,T_1^*,\ldots,T_n^*}}h_i(x)=1\mbox{ and }\sum_{i\in\mathcal{N}\setminus L_{x,T_1^*,\ldots,T_n^*}}h_i(x)=0\,,
        \]
    where
        \[
            L_{x,T_1^*,\ldots,T_n^*}:=\left\{i\in\mathcal{N}:
                T_i^*(\Pr(S>\underline{s}+x))=\max_{j\in\mathcal{N}}\{T_j^*(\Pr(S>\underline{s}+x))\}\right\}\,.
        \]
\end{corollary}

\begin{proof}
    By Corollary \ref{cor:po_ac_cpo}, the sets $\mathcal{CPO}$ and $\mathcal{PO}$ coincide. It then follows that the characterization in Theorem \ref{thm:cpo_coherent} applies to all Pareto-optimal allocations.
\end{proof}

\medskip
\subsection{An Algorithmic Approach to Finding Pareto Optima}
\label{subsec:algorithm}

The main advantage of the characterization in Theorem \ref{thm:cpo_coherent} is that it describes an algorithm that can be implemented to find the shape of comonotone Pareto optima. First, problem \eqref{eq:inf_problem} must be solved. This problem can be interpreted in the following manner:

\smallskip

\begin{itemize}
    \item Define the function $\psi:\prod_{i=1}^n\mathcal{T}_i\times\R_+\to\R$ by
    $$\psi\left((T_1,\ldots,T_n),x\right) := \max_{i\in\mathcal{N}}\,\left\{T_i(\Pr(S>\underline{s}+x))\right\}.$$
        That is, for a given vector of distortion functions $(T_1,\ldots,T_n)$ and a positive real number $x$, the function $\psi$ provides the most optimistic assessment of the likelihood of the tail event $\Pr(S>\underline{s}+x)$ among all agents.

\medskip

    \item Define the function $\Psi:\prod_{i=1}^n\mathcal{T}_i\to\R$ by
    $$\Psi\left(T_1,\ldots,T_n\right) := \int_0^\infty\psi(T_1,\ldots,T_n,x)\,dx.$$
        The function $\Psi$ provides an aggregate measure of the most optimistic assessment over all possible tail event likelihoods.

        \medskip
        
    \item Problem \eqref{eq:inf_problem} minimizes the function $\Psi$ over all possible choices of distortion functions. In other words, a solution to \eqref{eq:inf_problem} represents the worst case scenario of the most optimistic assessment of tail event likelihood.

\medskip

\item The shape of a Pareto-optimal comonotone allocation can now be explicitly determined through the following steps:

\smallskip

\begin{itemize}

    \item Determine all possible solutions to \eqref{eq:inf_problem}. 

    \medskip
    
    \item For each solution $(T_1^*,\ldots,T_n^*)$ to \eqref{eq:inf_problem} and each $x\in\R_+$, define the set
        \[
            L_{x,T_1^*,\ldots,T_n^*}:=\left\{i\in\mathcal{N}:
                T_i^*(\Pr(S>\underline{s}+x))
                =\max_{j\in\mathcal{N}}\{T_j^*(\Pr(S>\underline{s}+x))\}\right\}\,.
        \]
        The set $L_{x,T_1^*,\ldots,T_n^*}$ then represents the set of agents who have the most optimistic assessment of the likelihood of the tail event $\Pr(S>\underline{s}+x)$, when the preferences of each agent $i\in\mathcal{N}$ is  fixed as the distortion risk measure with respect to $T_i^*$.

        \medskip

    \item For each $i\in\mathcal{N}$, choose a function $h_i:\R_+\to[0,1]$ such that for every solution $(T_1^*,\ldots,T_n^*)$ to \eqref{eq:inf_problem} and almost every $x\in\R_+$,
        \begin{equation}
            \label{eq:h_condition}
            \sum_{i\in L_{x,T_1^*,\ldots,T_n^*}}h_i(x)=1 \ \text{ and }\sum_{i\in\mathcal{N}\setminus L_{x,T_1^*,\ldots,T_n^*}}h_i(x)=0\,.
        \end{equation}
        The functions $h_i$ represent the post-exchange marginal endowment of each agent $i$.

        \medskip
        
    \item Let $g_i^*(x):=\displaystyle\int_0^xh_i(z)\,dz$, and take
        \[
            Y_i^*:=g_i^*(S-\underline{s})+c_i^*\,,
        \]
        where $\{c_i^*\}_{i=1}^n\in\R^n$ is chosen such that $\underset{i\in\mathcal{N}}{\sum}c_i^*=\underline{s}$ and $\{Y_i^*\}_{i=1}^n\in\mathcal{IR}$.

        \medskip

    \item Reiterate this process for every choice of $h_i$ that satisfies \eqref{eq:h_condition}. The allocation that gives the highest value of aggregate welfare
        \[
            \sum_{i=1}^nU_i(Y_i^*)
        \]
    must be an element of $\mathcal{CS}_\One$, and hence comonotone Pareto optimal.
\end{itemize}    

\end{itemize}

\medskip

In the above algorithm, the only step that does not contain an explicit formula is the minimization of \eqref{eq:inf_problem} itself. Indeed, a closed-form expression for this problem is difficult to obtain. However, the advantage of \eqref{eq:inf_problem} is that the domain of the minimization is a product of sets of distortion functions. In many practical scenarios, this is a more tractable domain for numerical optimization, unlike the domain of all comonotone allocations, as is the case with the original sup-convolution problem $\mathcal{CS}_\One$. This advantage is illustrated through a numerical simulation in Section \ref{sec:risk_sharing}.

\medskip

Finally, note that the form imposed by Theorem \ref{thm:cpo_coherent} only provides a \emph{necessary} condition for Pareto optimality. This condition is \emph{not sufficient} in general, and we provide a counterexample in Appendix \ref{sec:not_sufficient_example}. Therefore, identifying the Pareto-optimal allocations among those suggested by Theorem \ref{thm:cpo_coherent} requires evaluating the utility of each candidate allocation, as described in the final step of the above algorithm. Nonetheless, in practice, the condition imposed on the marginal endowments \eqref{eq:h_condition} often specifies a narrow range of candidate allocations, from which the Pareto-optimal allocations can be easily identified. Our numerical example in Section \ref{sec:risk_sharing} provides one such example, where both the solution to \eqref{eq:inf_problem} and the allocation satisfying \eqref{eq:h_condition} are unique.

\medskip
\subsection{Concave Dual Utilities}
\label{subsec:dual_utilities}

In the special case where each set $\mathcal{T}_i$ is a singleton, the solution to \eqref{eq:inf_problem} is immediate. In fact, these are precisely the dual utility functionals introduced by \cite{yaari}. We show below that if $U_i$ satisfies the conditions of Theorem \ref{thm:cpo_coherent} as well as the additional condition of \emph{comonotone additivity}, then $\mathcal{T}_i$ can be taken to be a singleton. In this case, it is possible to obtain a representation of $\mathcal{CPO}$ in closed form.

\begin{definition}
\label{defn:drm}
    A functional \:$U$ is a \emph{dual utility functional} if there exists a distortion function $T$ such that for all random variables $Z\in\mathcal{X}$,
        \[
            U(Z)=\int Z\,dT\circ\Pr\,.
        \]
\end{definition}

\begin{proposition}{\cite[Theorem 2]{yaari}}
\label{prop:yaari_risk_aversion}
    Yaari's dual utility functional $Z \mapsto \displaystyle\int Z\,dT\circ\Pr$ is concave and Schur-concave if and only if the distortion function $T$ is convex.
\end{proposition}

It is well known that concave dual utilities can be characterized as the set of functionals that satisfy some standard properties, including comonotone additivity:

\begin{definition}
    A map $U:\mathcal{X}\to\R$ is \emph{comonotone-additive} if for every pair of comonotone random variables $Z_1,Z_2\in\mathcal{X}$,
        \[
            U(Z_1+Z_2)=U(Z_1)+U(Z_2)\,.
        \]
\end{definition}

In particular, for monetary utilities, comonotone additivity implies positive homogeneity (e.g., \citealt[Lemma 4.83]{FollmerSchied2016}). The following result shows that given the additional property of comonotone additivity, the setting of Theorem \ref{thm:cpo_coherent} reduces precisely to the case of dual utilities.

\begin{proposition}
\label{prop:dual_utility_rep}
{\cite{Kusuoka2001}, \cite[Theorem 2.3]{JouiniSchachermayerTouzi2008}}
    A functional $U$ is a concave dual utility functional if and only if it is a law-invariant comonotone-additive monetary utility.
\end{proposition}

The following result assumes that for each $i\in\mathcal{N}$, the utility functional $U_i$ is a dual utility functional with respect to a convex distortion. That is, there exists a convex distortion function $T_i$ such that
    \[
        U_i(Z)=\int Z\,dT_i\circ\Pr,
    \]
for all random variables $Z\in\mathcal{X}$. In this case, we can fully characterize the set of solutions $\mathcal{CS}_\One$, and hence the set of all comonotone Pareto-optimal allocations $\mathcal{CPO}$.

\begin{corollary}
\label{cor:cpo_drm_characterization}
    Suppose that each $U_i$ is a concave dual utility functional. For each $x\in\R_+$, let
        \[
            L_x:=\left\{i\in\mathcal{N}:
                T_i(\Pr(S>\underline{s}+x))=\max_{j\in\mathcal{N}}\{T_j(\Pr(S>\underline{s}+x))\}\right\}\,,
        \]
    and let $L_x^C:=\mathcal{N}\setminus L_x$. For each $i\in\mathcal{N}$, let $h_i:\R_+\to[0,1]$ be a function such that for almost every $x$,
        \[
            \sum_{i\in L_x}h_i(x)=1\mbox{ and }
                \sum_{i\in L_x^C}h_i(x)=0\,,
        \]
    and let $g_i^*(x):=\int_0^xh_i(z)\,dz$. Let $\{c_i^*\}_{i=1}^n\in\R^n$ be chosen such that   $\sum_{i=1}^nc_i^*=\underline{s}$ and $\{g_i^*(S-\underline{s})+c_i^*\}_{i=1}^n\in\mathcal{IR}$. Then
        \[
            \sup_{\{Y_i\}_{i=1}^{n}\in\mathcal{IR}\cap\mathcal{A}_C}\left\{\sum_{i=1}^{n}U_i\left(Y_i\right)\right\}
            =\sum_{i=1}^{n}U_i\left(g_i^*(S-\underline{s})+c_i^*\right)
            =\underline{s}+\int_0^\infty\max_{i\in\mathcal{N}}\,\{T_i(\Pr(S>\underline{s}+x))\}\,dx\,,
        \]
        
    \noindent and therefore $\{g_i^*(S-\underline{s})+c_i^*\}_{i=1}^n$ is comonotone Pareto optimal. Furthermore, all solutions are of this form. That is, if $\{Y_i^*\}_{i=1}^n\in\mathcal{CS}_\One$, then
        \[
            Y_i^*=g_i^*(S-\underline{s})+c_i^*\,,
        \]
    where $\{c_i^*\}_{i=1}^n\in\R^n$ and $\{g_i^*\}_{i=1}^n\in\mathcal{G}$ are of the form given above.
\end{corollary}

\smallskip

\begin{proof}
    Note that each $U_i$ has the representation
        \[
            U_i(Z)=\inf_{\tilde{T}\in\mathcal{T}_i}\int Z\,d\tilde{T}\circ\Pr\,,
        \]
    where for each $i$, the set $\mathcal{T}_i$ is the singleton set $\{T_i\}$. The result then follows from the proof of Theorem \ref{thm:cpo_coherent}.
\end{proof}

\subsection{Competitive Equilibria}
\label{subsec:equilibrium}

We now consider equilibria in a competitive market where each agent's preference is represented by a law-invariant positively homogeneous monetary utility. Let $\mathcal{P}$ denote the set of all probability measures on $(\Omega,\mathcal{F})$ that are absolutely continuous with respect to $\Pr$. We assume that the price of a random wealth $Z \in \mathcal{X}$ is given by the expectation of this random variable with respect to an element of $\mathcal{P}$, which we call the pricing measure. For a fixed pricing measure $\Q \in \mathcal{P}$, the demand problem faced by each agent $i\in\mathcal{N}$ is therefore

\begin{equation}
\label{eq:demand}
\underset{Y_i\in\mathcal{X}} \max \ U_i(Y_i)\quad s.t.\quad \E^\Q[Y_i]\le\E^\Q[X_i]\,.
\end{equation}

\noindent That is, each agent maximises their utility given the budgetary constraint that the price of their random payoff does not exceed the value of their initial endowment. The definition of a classical Arrow-Debreu equilibrium in this setting is the following.

\begin{definition}
    A pair $((Y_1^*,\ldots,Y_n^*),\Q^*)\in\mathcal{X}^n\times\mathcal{P}$ is \emph{an equilibrium} if
    \begin{enumerate}[label=(\roman*)]
        \item For each agent $i\in\mathcal{N}$, $Y_i^*$ solves the demand problem \eqref{eq:demand} under the pricing measure $\Q^*$.
        \item $\sum_{i=1}^nY_i^*=\sum_{i=1}^nX_i=S$. This condition is often referred to as the market-clearing condition.
    \end{enumerate}
\end{definition}

The following result established the existence of equilibria in our setting, as well as the first and second welfare theorems.

\begin{proposition}
\label{prop:equil_exists}
Suppose that for each $i\in\mathcal{N}$, the utility functional $U_i$ is a law-invariant and positively homogeneous monetary utility. Then there exists an equilibrium $((Y_1^*,\ldots,Y_n^*),\Q^*)\in\mathcal{X}^n\times\mathcal{P}$. Furthermore,

\smallskip

\begin{enumerate}[label=(\roman*)]
\item If $((Y_1^*,\ldots,Y_n^*),\Q^*)$ is an equilibrium, then the allocation $(Y_1^*,\ldots,Y_n^*)$ is Pareto optimal.

\smallskip

\item If $(Y_1^*,\ldots,Y_n^*)$ is a Pareto-optimal allocation, then there exists a pricing measure $\Q^*\in\mathcal{P}$ such that the pair
$$\left(\left(Y_1^*+\E^{\Q^*}[X_1-Y_1^*],\ldots,Y_n^*+\E^{\Q^*}[X_n-Y_n^*]\right),\Q^*\right)$$
is an equilibrium.
\end{enumerate}
\end{proposition}

\medskip

Note that when preferences are translation invariant, Pareto-optimal allocations are only unique up to translation. Specifically, if $(Y_1,\ldots,Y_n^*)$ is a Pareto-optimal allocation and $\{b_i\}_{i=1}^n\in\R^n$ satisfies $\sum_{i=1}^nb_i=0$, then the allocation $(Y_1^*+b_1,\ldots,Y_n^*+b_n)$ is also Pareto optimal. However, for a given pricing measure $\Q^*$, only one of these translations can possibly lead to an equilibrium allocation under the pricing measure $\Q^*$. This allocation is precisely what is given in part (ii) of Proposition \ref{prop:equil_exists}, as shown by the following result.

\begin{lemma}\label{LemmaEqu}
Let $(Y_1^*,\ldots,Y_n^*)$ be a Pareto-optimal allocation, and suppose that $\Q^*$ is a pricing measure such that $((Y_1^* + b_1, \ldots, Y_n^* + b_n),\Q^*)$ is an equilibrium, where $\{b_i\}_{i=1}^n\in\R^n$ satisfies $\sum_{i=1}^n b_i=0$. Then $b_i=\E^{\Q^*}[X_i-Y_i^*]$, for all $i\in\mathcal{N}$.
\end{lemma}

However, a more explicit characterization of the equilibrium pricing measure is difficult to obtain in general. In the special case of Yaari utilities, the pricing measure can be identified in terms of distortion functions, as shown in Section 3.1 of \cite{boonen2021competitive}, for instance. Their characterization result is given below.

\begin{proposition}
\label{prop:equil_yaari}
Suppose that for each $i\in\mathcal{N}$, the utility functional $U_i$ is a dual utility with respect to a convex distortion function $T_i$. Let $(Y_1^*,\ldots,Y_n^*)\in\mathcal{CPO}$. If $((Y_1^*+\E^{\Q^*}[X_1-Y_1^*],\ldots,Y_n^*+\E^{\Q^*}[X_n-Y_n^*]),\Q^*)$ is an equilibrium, then the pricing measure $\Q^*$ must satisfy
    \[
        T_{(n-1)}(\Pr(S>\underline{s}+x))\le\Q^*(S>\underline{s}+x)\le T_{(n)}(\Pr(S>\underline{s}+x)), \ \forall \, x\in\R_+\,,
    \]

\noindent where for each $t\in[0,1]$, $T_{(k)}(t)$ denotes the $k$-th smallest value of $T_i(t)$ over all $i\in\mathcal{N}$.
\end{proposition}

\medskip

For completeness, we provide a self-contained proof of Proposition \ref{prop:equil_yaari} in Appendix \ref{sec:proofs_equil}. However, this result does not easily extend to the more general setting of law-invariant positively homogeneous monetary utilities, and we leave the problem of further characterization of equilibrium pricing measures for future research.

\bigskip
\section{Pareto Optima in Risk-Sharing Markets}
\label{sec:risk_sharing}

In this section, we provide an application to risk-sharing markets. Specifically, we consider a risk-sharing market with $n$ agents, each of which measuring their risk exposure via a coherent risk measure, that is, a map $\rho:\mathcal{X}\to\R$ such that $-\rho$ is monotone, positively homogeneous, concave, and translation-invariant (e.g., \cite{artzner1999coherent}). Hence, coherent risk measures and positively homogeneous monetary utilities are equivalent up to a change in sign. Our previous results therefore apply to the setting of risk sharing considered in this section.

For each $i\in\mathcal{N}$, let $X_i\in\mathcal{X}$ denote the initial risk exposure of the $i$-th agent, which is assumed to be non-negative. In this section, for notational convenience, we use the convention that positive values of $X_i$ denote positive values of liabilities or risk exposures (i.e., negative values of endowments). The aggregate risk in the market is $S:=\sum_{i=1}^nX_i$. All agents are assumed to participate in a risk-sharing pool, which allocates the aggregate risk $S$ among the $n$ agents. The risk distributed to agent $i$ is denoted by $Y_i$.

For risk-sharing markets, individual rationality and Pareto optimality are defined in a similar manner to Definition \ref{defn:ir_po}, as shown below. Note that agents prefer smaller values of $\rho_i$ to larger ones.

\begin{definition}
    An allocation $\left\{Y^*_i\right\}_{i=1}^{n}\in\mathcal{A}$ is said to be:
    
    \begin{itemize}
    \item {\bf Individually Rational} (IR) if it incentivizes the parties to participate in the market, that is, 
        \begin{equation*}
            \rho_i\left(Y^*_i\right)\leq\rho_i\left(X_i\right),\quad \forall \, i \in \mathcal{N}.
        \end{equation*}
    \item {\bf Pareto Optimal} (PO) if it is IR and there does not exist any other IR allocation $\left\{Y_i\right\}_{i=1}^{n}$ such that
        \begin{equation*}
            \rho_i\left(Y_i\right)\leq\rho_i\left(Y^*_i\right),\quad \forall \, i \in \mathcal{N},
        \end{equation*}
    with at least one strict inequality.
    \end{itemize}
\end{definition}

\medskip
\subsection{Law-Invariant Coherent Risk Measures}

We assume that each agent $i \in \mathcal{N}$ uses a law-invariant coherent risk measure $\rho_i$. By Theorem \ref{thm:cpo_coherent}, we obtain a full characterization of all comonotone Pareto optima. This result is restated below in the context of risk sharing. Note that since we assume that the aggregate risk is non-negative, we may take the lower bound $\underline{s}=0$ to simplify the expressions.

\begin{corollary}
\label{cor:main_theorem_risk_sharing}
    Suppose that for each $i\in\mathcal{N}$, the risk measure $\rho_i$ is law-invariant and coherent. Let $\mathcal{T}_i$ denote the representing convex set of concave distortions for $\rho_i$ that is sequentially closed under pointwise convergence, as per Lemma \ref{lem:distortion_set_convex}. Then the following hold:

\smallskip
    
    \begin{enumerate}[label=(\roman*)]
        \item There exists a solution to the problem
            \begin{equation}
            \label{eq:inf_problem_risk_measures}
                \max_{\{T_i\}_{i=1}^n\in\,\prod_{i=1}^n\mathcal{T}_i}\,\int_0^\infty\min_{i\in\mathcal{N}}\,\{T_i(\Pr(S>x))\}\,dx\,.
            \end{equation}

            \smallskip
        
        \item A necessary condition for an allocation $\{Y_i^*\}_{i=1}^n$ to be comonotone Pareto optimal is that
                \[
                    Y_i^*=g_i^*(S)+c_i^*\,,
                \]

            \noindent where $\{c_i^*\}_{i=1}^n\in\R^n$ is chosen such that $\sum_{i=1}^nc_i^*=0$ and $\{g_i^*(S)+c_i^*\}_{i=1}^n\in\mathcal{IR}$, and $\{g_i^*\}_{i=1}^n\in\mathcal{G}$ can be written in terms of the integrals of suitable functions $h_i$.
            Specifically, for each $i\in\mathcal{N}$, we can write $g_i^*(x):=\int_0^xh_i(z)\,dz$, where each $h_i:\R_+\to[0,1]$ is a function such that for each $(T_1^*,\ldots,T_n^*)$ that solves \eqref{eq:inf_problem_risk_measures} and for almost every $x\in\R_+$,
                \[
                    \sum_{i\in L_{x,T_1^*,\ldots,T_n^*}}h_i(x)=1 \ \mbox{ and } \ \sum_{i\in\mathcal{N}\setminus L_{x,T_1^*,\ldots,T_n^*}}h_i(x)=0\,,
                \]
            where
                \begin{align*}
                    L_{x,T_1^*,\ldots,T_n^*}&:=\left\{i\in\mathcal{N}:
                        T_i^*(\Pr(S>x))=\min_{j\in\mathcal{N}}\{T_j^*(\Pr(S>x))\}\right\}\,.
                \end{align*}
    \end{enumerate}
    
\end{corollary}

The case of dual utilities, as examined in Subsection \ref{subsec:dual_utilities}, is analogous to the case where every agent in a risk-sharing market uses a coherent distortion risk measure. These risk measures are of the form
\begin{align*}
    \rho_i:\mathcal{X}&\to\R\\
    Z&\mapsto\int Z\,dT_i\circ\Pr\,,
\end{align*}
where $T_i$ is a concave distortion function. The result of Corollary \ref{cor:cpo_drm_characterization} applies in this case, and it is restated below in the setting of risk sharing. This characterization is known in the literature on distortion risk measures. See \cite[Theorem 3.3]{liu2020weighted} and \cite[Proposition 1]{boonen2021competitive}, for example.

\begin{corollary}
\label{cor:cpo_characterization_risk_version}
    Suppose that each agent uses a coherent distortion risk measure. For each $x\in\R_+$, let
        \[
            L_x:=\left\{i\in\mathcal{N}:
                T_i(\Pr(S>x))=\min_{j\in\mathcal{N}}\{T_j(\Pr(S>x))\}\right\}\,,
        \]
    and let $L_x^C:=\mathcal{N}\setminus L_x$. For each $i\in\mathcal{N}$, let $h_i:\R_+\to[0,1]$ be a function such that for almost every $x$,
        \[
            \sum_{i\in L_x}h_i(x)=1\mbox{ and }
                \sum_{i\in L_x^C}h_i(x)=0\,,
        \]
    and let $g_i^*(x):=\int_0^xh_i(z)\,dz$. Let $\{c_i^*\}_{i=1}^n\in\R^n$ be chosen such that $\sum_{i=1}^nc_i^*=0$ and $\{g_i^*(S)+c_i^*\}_{i=1}^n\in\mathcal{IR}$. Then 
        \[
            \inf_{\{Y_i\}_{i=1}^{n}\in\mathcal{IR}\cap\mathcal{A}_C}\left\{\sum_{i=1}^{n}\rho_i\left(Y_i\right)\right\}
                =\sum_{i=1}^{n}\rho_i\left(g_i^*(S)+c_i^*\right)
                =\int_0^\infty\min_{i\in\mathcal{N}}\,\{T_i(\Pr(S>x))\}\,dx\,,
        \]
    and therefore $\{g_i^*(S)+c_i^*\}_{i=1}^n$ is comonotone Pareto optimal. Furthermore, all comonotone Pareto-optimal allocations are of this form. That is, if $\{Y_i^*\}_{i=1}^n$ is comonotone Pareto optimal, then
        \[
            Y_i^*=g_i^*(S)+c_i^*\,,
        \]
    where $\{c_i^*\}_{i=1}^n\in\R^n$ and $\{g_i^*\}_{i=1}^n\in\mathcal{G}$ are of the form given above.
\end{corollary}

We note that in an insurance context, the structure of the allocations $\{Y_i^*\}_{i=1}^n$ admit the following interpretation. For each $i\in\mathcal{N}$, we have $Y_i^*=g_i^*(S)+c_i^*$. Here, the function $g_i^*$ can be seen as the retained risk allocated to agent $i$, as a function of the value of the loss $S$. Since each $g_i^*$ is increasing and 1-Lipschitz, these retention functions satisfy the so-called \textit{no-sabotage condition} of \cite{CarlierDana2003b,CarlierDana2005a}, which guarantees that no agent has an incentive to misreport their actual realized loss. The constants $c_i^*$ can be interpreted as the fixed premia that each agent $i$ pays to participate in the risk-sharing scheme. In the following, we will refer to the functions $g_i^*$ as retention functions.

\medskip
\subsection{An Example -- A Risk Sharing Problem}

As an application of the explicit characterization of comonotone Pareto optima, we consider a risk-sharing market from the perspective of the risk-bearing agents. As is common in risk management, risk measures are used to determine the amount of capital that each agent must hold in reserve to offset future liabilities. That is, $\rho_i(Z)$ represents the amount of capital that agent $i$ must reserve when faced with the risk $Z$. It is in the agents' best interest to reach a risk-sharing arrangement that allows them to minimize the amount of capital reserve required. Each agent's capital reserve must meet two targets. First, the agent must adhere to international capital requirement reporting standards, which typically prescribe the same risk measure to be used for all participating agents in the market. Second, each agent also uses an internal capital model for risk management at the institutional level. These internal models can vary based on the differences among agents in management procedures, accounting practices, etc. 

We assume that every risk measure in this scenario is a law-invariant coherent risk measure, and that admissible allocations are constrained to be comonotone. The results of Section \ref{sec:comonotone} imply that this restriction does not negatively affect the total welfare gain that is possible from risk sharing. Furthermore, as argued by \cite{embrechts2018quantile}, comonotonicity is an important property in the context of risk-sharing arrangements since it eliminates the possibility for moral hazard among collaborative agents.

We first examine the Pareto-optimal allocations that arise from regulation. A prominent set of standards on capital reserves reporting are those that regulate the international banking sector, as set by the Basel Committee on Banking Supervision. Also known as Basel IV, these standards suggest that institutions report their Expected Shortfall calculated with the parameter $\alpha=2.5\%$. We recall the following standard definitions:

\begin{definition}
    The \emph{Value-at-Risk (VaR)} at level $\alpha\in(0,1)$ of a random variable $Z\in\mathcal{X}$ under the probability measure $\Pr$ is
        \[
            \mathrm{VaR}_\alpha^\Pr(Z)
                :=\inf_{t\in\R}\left\{\Pr(Z>t)\le\alpha\right\}\,.
        \]
\end{definition}

\begin{definition}
    The \emph{Expected Shortfall (ES)} at level $\alpha\in(0,1)$ of a random variable $Z\in\mathcal{X}$ under the probability measure $\Pr$ is
        \[
            \mathrm{ES}_\alpha^\Pr(Z)
                :=\frac{1}{\alpha}\int_0^\alpha\mathrm{VaR}_u^\Pr(Z)\,du\,.
        \]
\end{definition}

It is well known that $\mathrm{ES}_\alpha^\Pr$ is a coherent distortion risk measure, with the distortion function $T(t)=\min\{t/\alpha,1\}$ (e.g., \citealt[Section 2.6.3.2]{denuit2006actuarial}). However, if each agent uses only the standardized capital requirement, then the following result implies that while comonotone allocations are Pareto optimal, it is not possible to find a Pareto improvement beyond any comonotone allocation.

\begin{proposition}
\label{prop:same_rho}
    If each agent uses the same coherent distortion risk measure, then every individually rational comonotone allocation is Pareto optimal.
\end{proposition}

\begin{proof}
    Let $\{Y_i^*\}_{i=1}^n$ be any individually rational comonotone allocation. Then by Lemma \ref{lem:translation}, $Y_i^*=g_i^*(S)+c_i^*$ where $\{g_i^*\}_{i=1}^n\in\mathcal{G}$ and $\{c_i^*\}_{i=1}^n\in\R^n$. Since $\{g_i^*\}_{i=1}^n\in\mathcal{G}$, each $g_i^*$ is 1-Lipschitz and therefore absolutely continuous. Hence, for each $i$, there exists $h_i:\R_+\to[0,1]$ such that $g_i^*(x)=\int_0^xh_i(z)\,dz$, a.e. Since $\sum_{i=1}^ng_i^*(x)=x$, we also have $\sum_{i=1}^nh_i(x)=1$, a.e. Additionally, since each agent uses the same coherent distortion risk measure, it follows that $L_x=\mathcal{N}$, for all $x\in\R_+$. Thus,
        \[
            \sum_{i\in L_x}h_i(x)=\sum_{i=1}^nh_i(x)=1 \ \mbox{ and } \ 
                \sum_{i\in L_x^C}h_i(x)=0\,,
        \]
    satisfying the conditions of Corollary \ref{cor:cpo_characterization_risk_version}. It then follows that $\{Y_i^*\}_{i=1}^n$ is CPO and PO.
\end{proof}

The result of Proposition \ref{prop:same_rho} is not surprising, since distortion risk measures are comonotone-additive maps. We may interpret the objective in problem \eqref{eq:inf_problem_risk_measures} as a measure of aggregate post-transfer risk in the market. If $\rho=\rho_1=\ldots=\rho_n$, we have
    \[
        \inf_{\{Y_i\}_{i=1}^{n}\in\mathcal{IR}\cap\mathcal{A}_C}
            \left\{\sum_{i=1}^{n}\rho_i\left(Y_i\right)\right\}
            =\sum_{i=1}^n\rho_i(S)=\rho(S)\,.
    \]

\medskip

\subsubsection{A Numerical Illustration}

Suppose that there are $n=3$ agents in the market, each using a different coherent distortion risk measure for internal capital management. We assume that the aggregate risk $S$ follows a Gamma distribution with shape parameter $k=2$ and scale parameter $\theta=10$. The mean of this distribution is $20$, and the variance of this distribution is $200$. Its probability density function is shown in Figure \ref{fig:density}.

\begin{figure}[hbtp!]
	\centering
	\includegraphics[scale=0.7]{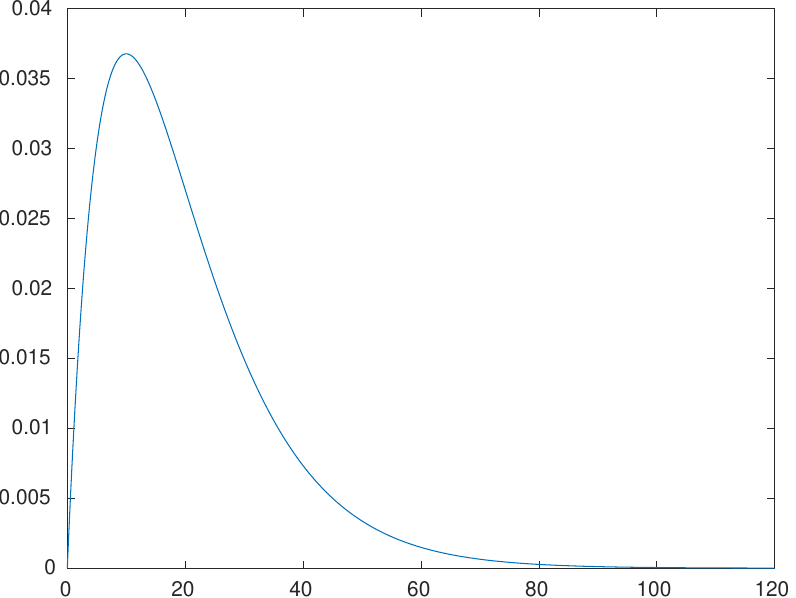}
	\caption{Probability Density Function of $S$.}
	\label{fig:density}
\end{figure}

We assume that all agents in this market are subject to the Basel IV regulatory standard of $\mathrm{ES}_{2.5\%}^\Pr$. This risk measure can be represented as a distortion risk measure, by using the distortion function
    \[
        \widehat{T}(t):=\min\{t/0.025,1\}\,.
    \]

\noindent However, each agent also possesses their individual capital requirements for internal usage. Agent 1 has decided that they would like to take a more conservative approach and reserve capital equal to the Expected Shortfall at the 1\% level. This is represented by the distortion function
    \[
        \widehat{T}_1(t):=\min\{t/0.01,1\}\,.
    \]

\noindent Since this agent must meet both the regulatory requirement and the internal capital requirement, their risk measure is therefore given by
    \begin{align*}
        \rho_1:\mathcal{X}&\to\R\\
        Z&\mapsto\max\left\{\int Z\,d\widehat{T}\circ\Pr,\int Z\,d\widehat{T}_1\circ\Pr\right\}\,.
    \end{align*}

To apply the result of Corollary \ref{cor:main_theorem_risk_sharing}, we need to express the risk measure $\rho_1$ in terms of a closed convex set of distortion functions. This is given in the following lemma, which provides the desired expression for each agent in this market. We use the notation for Agent $1$ in the statement and proof of this result for convenience, but the result applies to the other two agents as well.

\begin{lemma}
\label{lem:line_segment}
    Suppose that $\rho_1(Z)=\max\left\{\int Z\,d\widehat{T}\circ\Pr,\int Z\,d\widehat{T}_1\circ\Pr\right\}$, for each $Z \in \mathcal{X}$, and let $\mathcal{T}_1:=\left\{\lambda\widehat{T}+(1-\lambda)\widehat{T}_1 \,:\, \lambda\in[0,1]\right\}$ be the convex hull of the set $\left\{\widehat{T},\widehat{T}_1\right\}$. Then for all $Z\in\mathcal{X}$, we have 
        \[
            \rho_1(Z)=\sup_{T\in\mathcal{T}_1}\int Z\,dT\circ\Pr\,,
        \]
    and $\mathcal{T}_1$ is convex and sequentially closed under pointwise convergence.
\end{lemma}

\begin{proof}
The equality $\rho_1(Z)=\underset{T\in\mathcal{T}_1}{\sup}\displaystyle\int Z\,dT\circ\Pr$ follows from linearity of the Choquet integral in $T$. Furthermore, $\mathcal{T}_1$ is convex by definition. It remains to show that $\mathcal{T}_1$ is sequentially closed under pointwise convergence. To this end, let $\{\lambda_k\widehat{T}+(1-\lambda_k)\widehat{T}_1\}_{k=1}^\infty$ be a sequence in $\mathcal{T}_1$ that converges to a distortion function $\tilde{T}$ pointwise. Since $\lambda_k\in[0,1]$ for all $k$, the sequence $\{\lambda_k\}_{k=1}^\infty$ admits a converging subsequence $\{\lambda_{k_l}\}_{l=1}^\infty$ with limit $\lambda\in[0,1]$. It then follows that for each $t\in[0,1]$, we have
        \[
            \lim_{l\to\infty}\left\{\lambda_{k_l}\widehat{T}(t)+(1-\lambda_{k_l})\widehat{T}_1(t)\right\}
                =\lambda\widehat{T}(t)+(1-\lambda)\widehat{T}_1(t)
                =\tilde{T}(t)\,.
        \]
    Hence, $\tilde{T}=\lambda\widehat{T}+(1-\lambda)\widehat{T}_1\in\mathcal{T}_1$.
\end{proof}

\smallskip

We construct the risk measures for Agents 2 and 3 in a similar manner. Note that in the particular case of Agent 1, the Expected Shortfall at a level of $1\%$ is always greater than that at a level of $2.5\%$, and so the agent's risk measure can be instead just represented as $\mathrm{ES}_{1\%}^\Pr$. However, this will not be the case for the other agents.

\medskip

We assume that Agent $2$ also wishes to meet the requirement given by the distortion function
    \[
        \widehat{T}_2(t):=\min\{(t/0.05)^{0.3},1\}\,.
    \]

\noindent This is similar to the Expected Shortfall at a level of $5\%$, but assigns more weight to the extreme tail risk. Note that $\widehat{T}_2$ is concave. Let $\mathcal{T}_2$ be the convex hull of the set $\left\{\widehat{T}, \widehat{T}_2\right\}$. By Lemma \ref{lem:line_segment}, the risk measure of Agent 2 is represented by the set $\mathcal{T}_2$, in the sense that 
$$\rho_2(Z)=\underset{T\in\mathcal{T}_2}{\sup}\displaystyle\int Z\,dT\circ\Pr, \ \forall Z\in\mathcal{X}.$$

\smallskip

Finally, Agent 3 chooses to implement a distortion function given by Wang's transform (see \cite{wang2000class}) as follows:
    \[
    \widehat{T}_3(t):=\Phi(\Phi^{-1}(t)+2.8)\,,
    \]

\noindent where $\Phi$ denotes the distribution function of the standard normal distribution. Similarly, we let $\mathcal{T}_3$ denote the convex hull of the set $\left\{\widehat{T}, \widehat{T}_3\right\}$. A comparison of all distortion functions used in this scenario is illustrated in Figure \ref{fig:distortions}.

\smallskip

\begin{figure}[hbtp!]
	\centering
	\includegraphics[scale=0.7]{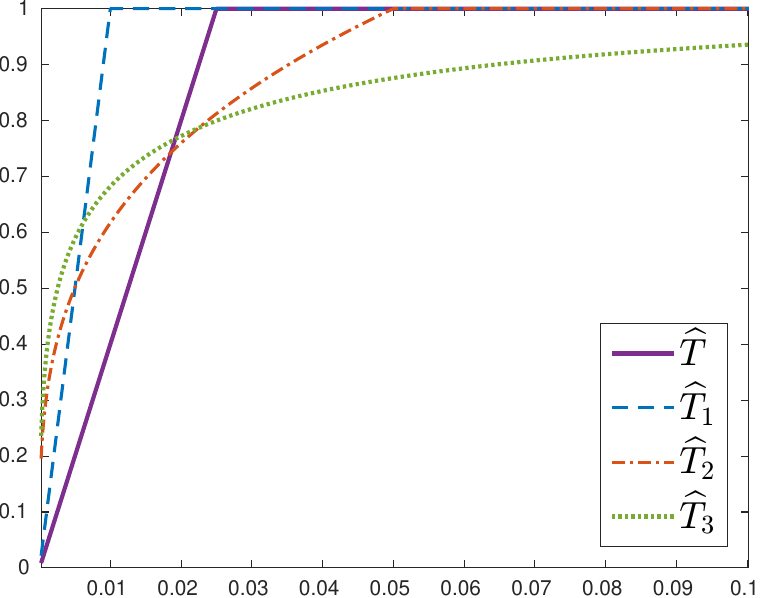}
	\caption{Probability Distortion Functions.}
	\label{fig:distortions}
\end{figure}

\smallskip

We can now apply Corollary \ref{cor:main_theorem_risk_sharing} to solve for comonotone Pareto optima in this risk-sharing market. 
First, it is determined through numerical optimization that the vector 
    \[
        (T_1^*, \, T_2^*, \, T_3^*) := (\widehat{T}_1, \, \widehat{T}_2, \, 0.2269 \, \widehat{T} + 0.7731 \, \widehat{T}_3)
    \]
is the unique solution to \eqref{eq:inf_problem_risk_measures}.
This yields the retention functions shown in Figure \ref{fig:retentions}. Furthermore, the optimal retention functions are unique in this case. Indeed, the marginal retention for each layer of the risk $S>x$ depends only on the agents with the lowest value of $T_i^*(\Pr(S>x))$, that is, the most optimistic likelihood assessment of the risk layer $S>x$. The agents' assessments of the likelihoods of tail events are shown in Figure \ref{fig:tstar}. 

\medskip

Now, recall from Corollary \ref{cor:main_theorem_risk_sharing} that for each $x\in\R_+$, the set $L_{x,T_1^*,T_2^*,T_3^*}$ denotes the set of agents with the most optimistic likelihood assessment of the tail event $S>x$, when the preference of agent $i$ is represented by a distortion risk measure with respect to $T_i^*$. It can be numerically verified that the agents with the most optimistic view towards the likelihood of such tail events are as follows:
    \[
        L_{x,T_1^*,T_2^*,T_3^*}=\begin{cases}
                \{3\}\,,&x\in(0,53.302)\cup(68.164,74.287)\\
                \{2\}\,,&x\in(53.302,68.164)\\
                \{1\}\,,&x>74.287
            \end{cases}\,.
    \]

\noindent Figure \ref{fig:retentions} shows how the retention function for agent $i$ increases with a slope of $1$ whenever $L_{x,T_1^*,T_2^*,T_3^*}=\{i\}$, that is, agent $i$ absorbs that tranche of the aggregate risk. Since $L_{x,T_1^*,T_2^*,T_3^*}$ is a singleton for a.e.\ $x$, it follows that each marginal retention is a.e.\ unique, which implies that retention functions are unique. That is, the retention structure shown in Figure \ref{fig:retentions} is the only possible structure that is comonotone Pareto optimal. Note that the retention is only increasing for an agent when that agent is most optimistic about the likelihood of the tail risk. 

\begin{figure}[htbp!]
\centering      
  \subfloat[Assessments of Tail Risks $T_i^*(\Pr(S>x))$.]{\includegraphics[scale=0.57]{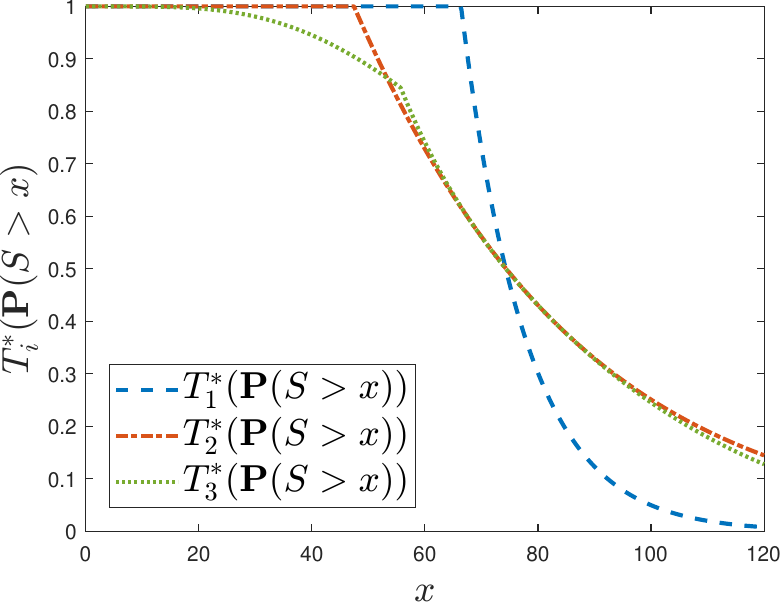}\label{fig:tstar}}
  \hspace{0.2cm}
  \subfloat[Pareto-Optimal Retention Structure.]{\includegraphics[scale=0.57]{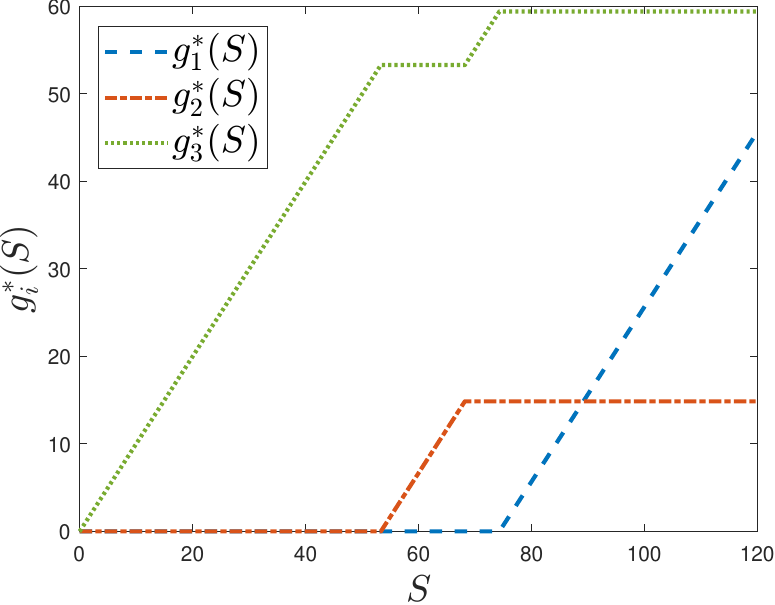}\label{fig:retentions}}
  \vspace{-0.2cm}
\caption{\vspace{0.6cm}}
\end{figure}

\bigskip
\section{Conclusion}
\label{sec:conclusion}

In this paper, we provide a characterization of Pareto-optimal allocations in a pure-exchange economy, in which agents have preferences represented by law-invariant positively homogeneous monetary utilities. Since these preferences are consistent with respect to the concave order, the classical comonotone improvement result applies. Identifying Pareto-optimal allocations then reduces to solving a sup-convolution problem over the set of comonotone allocations. By combining this result with duality representations of these preferences, we obtain a new characterization of Pareto-optimal allocations. The main advantage of our result is that provides explicit formulas for the shape of the comonotone optimal allocations themselves, as well as a clear and easily implementable algorithm for finding these optima.

\medskip

Our results may also be interpreted in the context of risk sharing, where these preferences are more commonly known as law-invariant coherent risk measures. This is a wide class of risk measures that encompasses many popular functionals in the insurance and risk management literature, including the expected shortfall. As a special case of our characterization result, we recover an explicit characterization of Pareto-optimal allocations when agents use law-invariant coherent risk measures. We apply this characterization to a problem of risk sharing in a numerical illustration.

\clearpage
\newpage
\hypertarget{LinkToAppendix}{\ }
\appendix

\vspace{-0.8cm}

\section{Proofs of the Main Results}
\label{appendix_proofs}

\subsection{Proof of Proposition \ref{prop:weighted_infconv}}

Suppose that $\{Y_i^*\}_{i=1}^n\in\mathcal{A}$ is not PO, and hence not weakly PO under our assumptions. Then by Lemma \ref{lem:po_strict}, there exists another IR allocation $\{Y_i\}_{i=1}^n$ such that
    \[
        U_i(Y_i)>U_i(Y_i^*),
            \ \forall\,i\in\mathcal{N}\,.
    \]
Therefore, since $\lambda\ge0$ and not all entries of $\lambda$ are zero,
    \[
        \sum_{i=1}^n\lambda_i\,U_i(Y_i)
            >\sum_{i=1}^n\lambda_i\,U_i(Y_i^*)\,.
    \]
Conversely, suppose $\{Y_i^*\}_{i=1}^n$ is PO. Define subsets of $\R^n$ by the following:
    \begin{align*}
        \mathcal{U}&:=co\,\Big\{\{U_i(Y_i)\}_{i=1}^n:\{Y_i\}_{i=1}^n\in\mathcal{IR}\Big\}\,,\\
        \mathcal{U}^-&:=\{r\in\R^n:r\le u\mbox{ for some }u\in\mathcal{U}\}\,,\\
        \mathcal{V}&:=\left\{r\in\R^n:r_i\ge U_i(Y_i^*)\right\}\setminus\left\{U_i(Y_i^*)\right\}_{i=1}^n\,,
    \end{align*}

\noindent where $co$ denotes the convex hull. Then by construction, both $\mathcal{U}$ and $\mathcal{V}$ are convex sets, and it is easy to verify that $\mathcal{U}^-$ is convex.

\medskip

We claim that $\mathcal{U}^-\cap\mathcal{V}=\varnothing$. Suppose for the sake of contradiction that $(r_1,\ldots,r_n)\in\mathcal{U}^-\cap\mathcal{V}$. Then $(r_1,\ldots,r_n)$ is dominated by some convex combination of elements in $\mathcal{U}$. That is, for each $i\in\mathcal{N}$, we have
    \[
        r_i\le\sum_{k=1}^mt_k\,U_i\left(Y_i^{(k)}\right)\,,
    \]
where $\sum_{k=1}^mt_k=1$ and $\left\{Y_i^{(k)}\right\}_{i=1}^n\in\mathcal{IR}$ for all $k\in1,\ldots,m$. Recall that for each $i\in\mathcal{N}$, the initial endowment of agent $i$ is denoted by $X_i$. By concavity of each $U_i$, we have
    \begin{equation}
        \label{eq:convex_ir}
        U_i(X_i)
            =\sum_{k=1}^mt_k\,U_i(X_i)
            \le\sum_{k=1}^mt_k\,U_i\left(Y_i^{(k)}\right)
            \le U_i\left(\sum_{k=1}^mt_k\,Y_i^{(k)}\right)\,.
    \end{equation}
Furthermore,    
    \begin{equation}
        \label{eq:convexity}
        U_i(Y_i^*)
            \le r_i
            \le\sum_{k=1}^mt_k\,U_i\left(Y_i^{(k)}\right)
            \le U_i\left(\sum_{k=1}^mt_k\,Y_i^{(k)}\right)\,,
    \end{equation}

\noindent where the inequality is strict for some $i\in\mathcal{N}$. On the other hand,
    \[
        \sum_{i=1}^n\sum_{k=1}^mt_k\,Y_i^{(k)}
            =\sum_{k=1}^mt_k\sum_{i=1}^nY_i^{(k)}
            =\sum_{k=1}^mt_k\,S=S\,,
    \]
implying that $\left\{\sum_{k=1}^mt_k\,Y_i^{(k)}\right\}_{i=1}^n\in\mathcal{A}$, and hence $\left\{\sum_{k=1}^mt_k\,Y_i^{(k)}\right\}_{i=1}^n\in\mathcal{IR}$ by \eqref{eq:convex_ir}. Therefore $\left\{\sum_{k=1}^mt_k\,Y_i^{(k)}\right\}_{i=1}^n$ is an IR allocation that improves upon $\{Y_i^*\}_{i=1}^n$, which contradicts the assumption that $\{Y_i^*\}_{i=1}^n\in\mathcal{PO}$. We conclude that $\mathcal{U}^-\cap\mathcal{V}=\varnothing$.

\medskip

The hyperplane separation theorem \cite[Theorem 5.61]{AliprantisBorder} then implies that there exists $\lambda\in\R^n\setminus\{0\}$ such that
    \[
        \lambda\cdot u\le\lambda\cdot v,
        \ \forall\,(u,v)\in\mathcal{U}^-\times\mathcal{V}\,.
    \]
Since $\{U_i(Y_i^*)\}_{i=1}^n$ is a limit point of $\mathcal{V}$, the above implies
    \begin{equation}
        \label{eq:hyperplane}
        \lambda\cdot u\le\lambda\cdot\{U_i(Y_i^*)\}_{i=1}^n, 
            \ \forall\,u\in\mathcal{U}^-\,.
    \end{equation}

\noindent We now show that $\lambda\ge0$. To this end, for $i\in\mathcal{N}$, let $e_i\in\R^n$ be the vector with $1$ in its $i$-th coordinate, and $0$ elsewhere. Then since $\{U_i(Y_i^*)\}_{i=1}^n\in\mathcal{U}$, we have $\{U_i(Y_i^*)\}_{i=1}^n-e_j\in\mathcal{U}^-$ for any $j\in\mathcal{N}$. Substituting this into \eqref{eq:hyperplane} gives
\begin{equation*}
\lambda \cdot \{U_i(Y_i^*)\}_{i=1}^n -\lambda_j \leq \lambda \cdot \{U_i(Y_i^*)\}_{i=1}^n \ \ \hbox{and} \ \ \lambda_j \geq 0.
\end{equation*}

To complete the proof, note that \eqref{eq:hyperplane} implies that
    \[
        \sum_{i=1}^n\lambda_i\,U(Y_i^*)
            \ge\lambda\cdot u, \ \forall \,u \in \mathcal{U}\,.
    \]
However, \eqref{eq:convex_ir} and \eqref{eq:convexity} imply that for any $(u_1,\ldots,u_n)\in\mathcal{U}$, it is possible to find an allocation $\{Y_i\}_{i=1}^n\in\mathcal{IR}$ such that 
    \[
        U_i(Y_i)\ge u_i, \    \forall\,i\in\mathcal{N}\,.
    \]
Hence, we have
    \[
        \sum_{i=1}^n\lambda_i\,U(Y_i^*)
            \ge\sup_{\{Y_i\}_{i=1}^n\in\mathcal{IR}}\sum_{i=1}^n\lambda_i\,U(Y_i)\,,
    \]
and since $\{Y_i^*\}_{i=1}^n$ is feasible for this problem, it must be a solution.\qed

\bigskip
\subsection{Proof of Theorem \ref{thm:PO_vs_CPO}}

By Proposition \ref{prop:weighted_infconv} and Corollary \ref{cor:comonotone_characterization}, we have $\mathcal{PO}=\cup_{\lambda\in\Lambda}\mathcal{S}_\lambda$ and $\mathcal{CPO}=\cup_{\lambda\in\Lambda}\mathcal{CS}_\lambda$. For a given $\lambda\in\Lambda$, we will first show that if $\mathcal{S}_\lambda$ is non-empty, then $\mathcal{CS}_\lambda$ is also non-empty and \eqref{eq:inf_convs_equal} holds. To conclude the proof, we show that $\mathcal{CS}_\lambda\ne\varnothing$ implies $\mathcal{S}_\lambda\ne\varnothing$.

\medskip

Suppose that $\mathcal{S}_\lambda$ is non-empty, and let $\{Y_i^*\}_{i=1}^n\in\mathcal{S}_\lambda$. Then it is immediate that
    \begin{align*}        
        \sum_{i=1}^{n}\lambda_i\,U_i\left(Y_i^*\right)&=\sup_{\{Y_i\}_{i=1}^{n}\in\mathcal{IR}}\left\{\sum_{i=1}^{n}\lambda_i\,U_i(Y_i)\right\}
        \ge\sup_{\{Y_i\}_{i=1}^{n}\in\mathcal{IR}\cap\mathcal{A}_C}\left\{ \sum_{i=1}^{n}\lambda_i\,U_i(Y_i)\right\}\,.
    \end{align*}
By Proposition \ref{prop:comonotone_improvement}, there exists a comonotone allocation $\left\{\tilde{Y}_i\right\}_{i=1}^n\in\mathcal{A}_C$ such that
    \[
        Y_i^*\ccv\tilde{Y}_i, \ 
            \forall i\in\{1,\dots,n\}\,.
    \]
Since each $U_i$ is Schur concave, this implies that
    \[
        U_i\left(\tilde{Y}_i\right)\ge U_i(Y_i^*),
            \ \forall i\in\{1,\ldots,n\}\,.
    \]
However, since $\{Y_i^*\}_{i=1}^n$ is Pareto optimal, equality must hold for each $i$. That is,
    \[
U_i\left(\tilde{Y}_i\right)=U_i(Y_i^*),
            \ \forall i\in\{1,\ldots,n\}\,,
    \]
implying that $\left\{\tilde{Y}_i\right\}_{i=1}^n\in\mathcal{IR}\cap\mathcal{A}_C$. Hence,
    \begin{align*}
        \sup_{\{Y_i\}_{i=1}^{n}\in\mathcal{IR}\cap\mathcal{A}_C}
            \left\{\sum_{i=1}^{n}\lambda_i\,U_i(Y_i)\right\}
            &\ge\sum_{i=1}^{n}\lambda_i\,U_i\left(\tilde{Y}_i\right)
        =\sup_{\{Y_i\}_{i=1}^{n}\in\mathcal{IR}}
            \left\{\sum_{i=1}^{n}\lambda_i\,U_i(Y_i)\right\}\,,
    \end{align*}
    
\noindent so \eqref{eq:inf_convs_equal} holds.

\medskip

It remains to show that if $\mathcal{CS}_\lambda$ is non-empty, then so is $\mathcal{S}_\lambda$. Let $\{Y_i^*\}_{i=1}^n\in\mathcal{CS}_\lambda$. It is immediate that
    \[ 
        \sup_{\{Y_i\}_{i=1}^{n}\in\mathcal{IR}}
            \left\{\sum_{i=1}^{n}\lambda_i\,U_i(Y_i)\right\}
            \ge\sup_{\{Y_i\}_{i=1}^{n}\in\mathcal{IR}\cap\mathcal{A}_C}
            \left\{\sum_{i=1}^{n}\lambda_i\,U_i(Y_i)\right\}\,.
    \]
Suppose for the sake of contradiction that this inequality is strict. That is, there exists an allocation $\{\hat{Y}_i\}_{i=1}^n\in\mathcal{IR}$ such that
    \begin{align*}
        \sum_{i=1}^{n}\lambda_i\,U_i\left(\hat{Y}_i\right)
            &>\sup_{\{Y_i\}_{i=1}^{n}\in\mathcal{IR}\cap\mathcal{A}_C}
            \left\{\sum_{i=1}^{n}\lambda_i\,U_i(Y_i)\right\}
        =\sum_{i=1}^{n}\lambda_i\,U_i(Y_i^*)\,.
    \end{align*}
Then by Proposition \ref{prop:comonotone_improvement}, there exists a comonotone allocation $\left\{\tilde{Y}_i\right\}_{i=1}^n\in\mathcal{A}_C$ such that
\[U_i\left(\tilde{Y}_i\right)\ge U_i\left(\hat{Y}_i\right), \ \forall i\in\{1,\ldots,n\}\,,\]
which implies that $\left\{\tilde{Y}_i\right\}_{i=1}^n$ is individually rational as well. Therefore
    \begin{align*}
        \sup_{\{Y_i\}_{i=1}^{n}\in\mathcal{IR}\cap\mathcal{A}_C}
            \left\{\sum_{i=1}^{n}\lambda_i\,U_i\left(Y_i\right)\right\}
            &\ge\sum_{i=1}^{n}\lambda_i\,U_i\left(\tilde{Y}_i\right)
        \ge\sum_{i=1}^{n}\lambda_i\,U_i\left(\hat{Y}_i\right)
        >\sum_{i=1}^{n}\lambda_i\,U_i\left(Y_i^*\right)\\
        &=\sup_{\{Y_i\}_{i=1}^{n}\in\mathcal{IR}\cap\mathcal{A}_C}
            \left\{\sum_{i=1}^{n}\lambda_i\,U_i\left(Y_i\right)\right\}\,,
    \end{align*}

\noindent a contradiction. Hence, \eqref{eq:inf_convs_equal} holds, and $\{Y_i^*\}_{i=1}^n\in\mathcal{S}_\lambda$, implying that $\mathcal{S}_\lambda\ne\varnothing$, as desired.\qed

\bigskip
\subsection{Proof of Proposition \ref{prop:characterization}}

If $\left\{Y^*_i\right\}_{i=1}^{n}\not\in\mathcal{PO}$, then there exists $\left\{\tilde{Y}_i\right\}_{i=1}^{n}\in\mathcal{IR}$ such that
    \[
        U_i\left(\tilde{Y}_i\right)
            \ge U_i\left(Y^*_i\right),
    \]
with at least one strict inequality, which implies that
    \[
        \sum_{i=1}^{n}U_i\left(\tilde{Y}_i\right)
>\sum_{i=1}^{n}U_i\left(Y^*_i\right).
    \]

\noindent Therefore, $\left\{Y^*_i\right\}_{i=1}^{n}\not\in\mathcal{S}_\One$, and hence $\mathcal{S}_1\subseteq\mathcal{PO}$.

\medskip

To show the reverse inclusion, assume, by way of contradiction, that there exists $\left\{Y^*_i\right\}_{i=1}^{n}\in\mathcal{PO}$ such that $\left\{Y^*_i\right\}_{i=1}^{n}\notin\mathcal{S}_\One$. Then, there exists $\left\{\tilde{Y}_i\right\}_{i=1}^{n}\in\mathcal{IR}$ such that
    \begin{equation}
    \label{eq:smaller_sum}
        \sum_{i=1}^{n}U_i\left(\tilde{Y}_i\right)
            >\sum_{i=1}^{n}U_i\left(Y^*_i\right)\,.
    \end{equation}
Define $\mathcal{N}_1,\mathcal{N}_2,\mathcal{N}_3\subseteq\mathcal{N}$ such that,
    \begin{equation*}
        U_i\left(\tilde{Y}_i\right)
            <U_i\left(Y^*_i\right),\quad\forall i\in\mathcal{N}_1\,,
    \end{equation*}
    \begin{equation*}
        U_i\left(\tilde{Y}_i\right)
            =U_i\left(Y^*_i\right),\quad\forall i\in\mathcal{N}_2\,,
    \end{equation*}
    \begin{equation*}
        U_i\left(\tilde{Y}_i\right)
            >U_i\left(Y^*_i\right),\quad\forall i\in\mathcal{N}_3\,.
    \end{equation*}

\noindent Note that $\mathcal{N}_1,\mathcal{N}_2,\mathcal{N}_3$ is a partition of $\mathcal{N}$. Moreover, by \eqref{eq:smaller_sum}, $\mathcal{N}_3\neq\varnothing$.

\medskip

By assumption, $\left\{Y^*_i\right\}_{i=1}^{n}\in\mathcal{PO}$, which implies that $\mathcal{N}_1\neq\varnothing$. Define, for $i\in\mathcal{N}_1$,
    \begin{equation*}
        \varepsilon_i
            :=U_i\left(Y^*_i\right)-U_i\left(\tilde{Y}_i\right)
            >0\,.
    \end{equation*}
Then, by \eqref{eq:smaller_sum}, there exist $\left\{\varepsilon_i\right\}_{i\in\mathcal{N}_3}$ such that, (i) $\varepsilon_i\geq 0$, for $i\in\mathcal{N}_3$; (ii) $U_i\left(\tilde{Y}_i-\varepsilon_i\right)\geq U_i\left(Y^*_i\right)$, for $i\in\mathcal{N}_3$, with at least one strict inequality; and (iii) $\sum_{i\in\mathcal{N}_3}\varepsilon_i=\sum_{i\in\mathcal{N}_1}\varepsilon_i$. Let
    \[
        \hat{Y}_i:=\begin{cases}
            \,\tilde{Y}_i+\varepsilon_i\,,
                &\quad\forall i\in\mathcal{N}_1\\
            \,\tilde{Y}_i\,
                &\quad\forall i\in\mathcal{N}_2\\
            \,\tilde{Y}_i-\varepsilon_i\,,
                &\quad\forall i\in\mathcal{N}_3
            \end{cases}
        \,.
    \]
Note that $\left\{\hat{Y}_i\right\}_{i=1}^{n}\in\mathcal{A}$, since
    \begin{equation*}
        \sum_{i=1}^{n}\hat{Y}_i
            =\sum_{i=1}^{n}\tilde{Y}_i-\sum_{i\in\mathcal{N}_1}\varepsilon_i+\sum_{i\in\mathcal{N}_3}\varepsilon_i
            =\sum_{i=1}^{n}\tilde{Y}_i=\sum_{i=1}^{n}X_i\,.
    \end{equation*}
Moreover, $\left\{\hat{Y}_i\right\}_{i=1}^{n}\in\mathcal{IR}$, since
    \begin{align*}
U_i\left(\hat{Y}_i\right)&=U_i\left(\tilde{Y}_i\right)+\varepsilon_i
            =U_i\left(\tilde{Y}_i\right)
            +\left(U_i\left(\tilde{Y}_i\right)-U_i\left(Y^*_i\right)\right)\\
        &=U_i\left(Y^*_i\right)
            \geq U_i\left(X_i\right)\,,
            \quad\forall i\in\mathcal{N}_1\,,
    \end{align*}
    \begin{equation*}
U_i\left(\hat{Y}_i\right)=U_i\left(\tilde{Y}_i\right)
            =U_i\left(Y^*_i\right)
            \geq U_i\left(X_i\right)\,,
             \quad\forall i\in\mathcal{N}_2\,,
    \end{equation*}
    \begin{equation}
    \label{eq:strict_inequality_N_3}
        U_i\left(\hat{Y}_i\right)
            =U_i\left(\tilde{Y}_i-\varepsilon_i\right)\geq U_i\left(Y^*_i\right)
            \geq U_i\left(X_i\right)\,,
            \quad\forall i\in\mathcal{N}_3\,,
    \end{equation}
    
\noindent in which \eqref{eq:strict_inequality_N_3} has at least one strict inequality. Hence, $\left\{Y^*_i\right\}_{i=1}^{n}\not\in\mathcal{PO}$, a contradiction. \qed

\bigskip
\subsection{Proof of Lemma \ref{lem:distortion_set_convex}}

By Lemma \ref{lem:coherent_rep}, there exists a convex law-invariant set $\mathcal{H}\subseteq L_+^\infty$ such that
    \[
        \mathcal{U}(Z)=
            \inf_{T\in\{\phi_H:H\in\mathcal{H}\}}\int Z\,dT\circ\Pr\,.
    \]
We claim that defining $\mathcal{T}$ to be the pointwise closure of the set
    \[
        \{\phi_H:H\in\mathcal{H}\cap\mathcal{X}^\uparrow\}
    \]
yields the desired convex set of distortions. It is immediate that
    \[
        \inf_{T\in\{\phi_H:H\in\mathcal{H}\}}\int Z\,dT\circ\Pr
            \le\inf_{T\in\{\phi_H:H\in\mathcal{H}\cap\mathcal{X}^\uparrow\}}\int Z\,dT\circ\Pr\,.
    \]
Suppose for the sake of contradiction that there exists $\tilde{H}\in\mathcal{H}$ such that
    \[
        \int Z\,d\phi_{\tilde{H}}\circ\Pr<\inf_{T\in\{\phi_H:H\in\mathcal{H}\cap\mathcal{X}^\uparrow\}}\int Z\,dT\circ\Pr\,.
    \]
Then since the probability space is non-atomic and $\mathcal{H}$ is a law-invariant set, there exists $H^\uparrow\in\mathcal{H}\cap\mathcal{X}^\uparrow$ such that $H^\uparrow$ and $\tilde{H}$ have the same distribution  \cite[Lemma 4.2]{dana2005representation}. However, by the Hardy-Littlewood inequality \cite[Theorem A.28]{FollmerSchied2016}, we have
    \[
        \int Z\,d\phi_{H^\uparrow}\circ\Pr
            =\int_0^1F_Z^{-1}(t)F_{H^\uparrow}^{-1}(1-t)\,dt
            \le\int_0^1F_Z^{-1}(t)F_{\tilde{H}}^{-1}(1-t)\,dt
            =\int Z\,d\phi_{\tilde{H}}\circ\Pr\,,
    \]
a contradiction. Hence,
    \[
        \mathcal{U}(Z)=\inf_{T\in\{\phi_H:H\in\mathcal{H}\}}\int Z\,dT\circ\Pr
            =\inf_{T\in\{\phi_H:H\in\mathcal{H}\cap\mathcal{X}^\uparrow\}}\int Z\,dT\circ\Pr\,.
    \]
Next, we show that $\{\phi_H:H\in\mathcal{H}\cap\mathcal{X}^\uparrow\}$ is convex. To this end, let $\lambda\in(0,1)$ and 
$H_1,H_2\in\mathcal{H}\cap\mathcal{X}^\uparrow$. For every $x\in[0,1]$, we have
    \begin{align*}
        \lambda\phi_{H_1}(x)+(1-\lambda)\,\phi_{H_2}(x)
            &=\lambda\int_0^xF_{H_1}^{-1}(t)\,dt+(1-\lambda)\int_0^xF_{H_2}^{-1}(t)\,dt\\
        &=\int_0^x\lambda\,F_{H_1}^{-1}(t)+(1-\lambda)\,F_{H_2}^{-1}(t)\,dt
        =\int_0^xF_{\lambda\,H_1+(1-\lambda)\,H_2}^{-1}(t)\,dt\\
        &=\phi_{\lambda\,H_1+(1-\lambda)\,H_2}(x),
    \end{align*}

\noindent where the third equality holds since $H_1$ and $H_2$ are comonotone. Since $\mathcal{H}\cap\mathcal{X}^\uparrow$ is convex, we have $\lambda\,Y_1+(1-\lambda)\,Y_2\in\mathcal{H}\cap\mathcal{X}^\uparrow$, and so $\{\phi_H:H\in\mathcal{H}\cap\mathcal{X}^\uparrow\}$ is convex as well.

\medskip

Finally, we show that the Choquet integral $\int Z\,dT\circ\Pr$ is sequentially continuous in $T$ with respect to pointwise convergence. Let $\{T^{(k)}\}_{k=1}^\infty$ be a sequence of distortion functions such that for all $t\in[0,1]$, we have $\lim_{k\to\infty}T^{(k)}(t)=T(t)$, for some function $T$. It is immediate that if $\{T^{(k)}\}_{k=1}^\infty$ converges pointwise to $T$ and each $T^{(k)}$ is a convex distortion function, then $T$ is also a convex distortion function. Let $Z\in\mathcal{X}$, and let $K:=||Z||_\infty<\infty$. Then we have
    \begin{align*}
        \lim_{k\to\infty}\int Z\,dT^{(k)}\circ\Pr&=
            \lim_{k\to\infty}\int_0^\infty T^{(k)}(\Pr(Z>x))\,dx
            +\lim_{k\to\infty}\int_{-\infty}^0[T^{(k)}(\Pr(Z>x))-1]\,dx\\
        &=\lim_{k\to\infty}\int_0^KT^{(k)}(\Pr(Z>x))\,dx
            +\lim_{k\to\infty}\int_{-K}^0[T^{(k)}(\Pr(Z>x))-1]\,dx\\
        &=\int_0^K\lim_{k\to\infty}T^{(k)}(\Pr(Z>x))\,dx
            +\int_{-K}^0\left[\lim_{k\to\infty}T^{(k)}(\Pr(Z>x))-1\right]\,dx\\
        &=\int_0^\infty T(\Pr(Z>x))\,dx
            +\int_{-\infty}^0[T(\Pr(Z>x))-1]\,dx
        =\int Z\,dT\circ\Pr\,,
    \end{align*}
where we may apply the dominated convergence theorem to exchange the order of the limit and the integral, since for any $x\in\R$,
    \[
        |T(\Pr(Z>x))|\le1 
        \ \ \hbox{and} \ \ 
        |T(\Pr(Z>x))-1|\le1\,.
    \]
Hence, the Choquet integral is sequentially continuous in $T$. Therefore
    \[
        \mathcal{U}(Z)=\inf_{T\in\{\phi_H:H\in\mathcal{H}\cap\mathcal{X}^\uparrow\}}\int Z\,dT\circ\Pr
            =\inf_{T\in\mathcal{T}}\int Z\,dT\circ\Pr\,,
    \]
where $\mathcal{T}$ has the desired properties.\qed

\bigskip
\subsection{Proof of Theorem \ref{thm:cpo_coherent}}
\begin{enumerate}[label=(\roman*)]
    \item
    Since the objective function of \eqref{eq:inf_problem} is non-negative, we have
        \[
            -\infty<V:=\inf_{\{T_i\}_{i=1}^n\in\prod_{i=1}^n\mathcal{T}_i}\,\int_0^\infty\max_{i\in\mathcal{N}}\,\{T_i(\Pr(S>\underline{s}+x))\}\,dx\,.
        \]
    Let $\{(T_1^{(k)},\ldots,T_n^{(k)})\}_{k=1}^\infty$ be a sequence such that
        \[
            \int_0^\infty\max_{i\in\mathcal{N}}\,\{T_i^{(k)}(\Pr(S>\underline{s}+x))\}\,dx\le V+\frac{1}{k}\,.
        \]
    
    \noindent Then since each distortion is a monotone function on the bounded interval $[0,1]$, by Helly's compactness theorem (e.g., \citealt[pp.\ 165-166]{Doob94}), there exists a subsequence $\left\{\left(T_1^{(k_j)},\ldots,T_n^{(k_j)}\right)\right\}_{j=1}^\infty$ for which $\left\{T_1^{(k_j)}\right\}_{j=1}^\infty$ converges pointwise to a limit $T_1^*$. Since $\mathcal{T}_1$ is closed under pointwise convergence by Lemma \ref{lem:distortion_set_convex}, we have $T_1^*\in\mathcal{T}_1$. Applying Helly's compactness theorem again to this subsequence gives another sequence $\left\{(T_1^{(k_l)},\ldots,T_n^{(k_l)})\right\}_{l=1}^\infty$ such that $\left\{T_1^{(k_l)}\right\}_{l=1}^\infty$ converges pointwise to $T_1^*\in\mathcal{T}_1$ and $\left\{T_2^{(k_l)}\right\}_{l=1}^\infty$ converges pointwise to $T_2^*\in\mathcal{T}_2$. Hence, iterating this process $n$ times yields a subsequence $\left\{\left(T_1^{(k_m)},\ldots,T_n^{(k_m)}\right)\right\}_{m=1}^\infty$ that converges pointwise to a limit  
        \[
(T_1^*,\ldots,T_n^*) \, \in \, \prod_{i=1}^n\mathcal{T}_i\,.
        \]

Therefore,
\begin{align*}
V&\le\int_0^\infty\max_{i\in\mathcal{N}}\,\{T_i^*(\Pr(S>\underline{s}+x))\}\,dx
=\int_0^\infty\max_{i\in\mathcal{N}}\,\left\{\lim_{m\to\infty}T_i^{(k_m)}(\Pr(S>\underline{s}+x))\right\}\,dx\\
&=\lim_{m\to\infty}\int_0^\infty\max_{i\in\mathcal{N}}\,\left\{T_i^{(k_m)}(\Pr(S>\underline{s}+x))\right\}\,dx
            \le V+\lim_{m\to\infty}\frac{1}{k_m}=V,
        \end{align*}
    
    \noindent where we can exchange the limit and the integral by dominated convergence. Hence, $(T_1^*,\ldots,T_n^*)$ is a solution to \eqref{eq:inf_problem}.
    
\medskip
    
    \item
    By Corollary \ref{cor:comonotone_characterization_translation_invar}, it suffices to characterize $\mathcal{CS}_\One$, i.e., solutions to the problem
        \[
            \sup_{\{Y_i\}_{i=1}^{n}\in\mathcal{IR}\cap\mathcal{A}_C}
                \left\{\sum_{i=1}^{n}U_i\left(Y_i\right)\right\}\,.
        \]
     By Lemma \ref{lem:translation}, this allocation can be written in terms of functions $\{g_i\}_{i=1}^n\in\mathcal{G}$ and constants $\{c_i\}_{i=1}^n\in\R^n$ where $\sum_{i=1}^nc_i=\underline{s}$. Conversely, if $\{g_i\}_{i=1}^n\in\mathcal{G}$ and $\{c_i\}_{i=1}^n\in\R^n$ with $\sum_{i=1}^nc_i=\underline{s}$, then $\{g_i(S-\underline{s})+c_i\}_{i=1}^n\in\mathcal{A}_C$. Therefore
        \begin{align*}
            \sup_{\{Y_i\}_{i=1}^{n}\in\mathcal{IR}\cap\mathcal{A}_C}
                \left\{\sum_{i=1}^{n}U_i\left(Y_i\right)\right\}
                &=\sup_{\substack{\left(\{g_i\}_{i=1}^{n},\{c_i\}_{i=1}^n\right)
                \in\mathcal{IR}\cap(\mathcal{G}\times\R^n)\\
                \sum_{i=1}^nc_i=\underline{s}
                }}
                \left\{\sum_{i=1}^{n}U_i\left(g_i(S-\underline{s})+c_i\right)\right\}\\
            &=\sup_{\left(\{g_i\}_{i=1}^{n},\{c_i\}_{i=1}^n\right)\in\mathcal{IR}\cap(\mathcal{G}\times\R^n)}
                \left\{\sum_{i=1}^{n}U_i\left(g_i(S-\underline{s})\right)\right\}+\underline{s}\,,
        \end{align*}
    
    \noindent where we write $(\{g_i\}_{i=1}^n,\{c_i\}_{i=1}^n)\in\mathcal{IR}$ when the allocation $\{g_i(S-\underline{s})+c_i\}_{i=1}^n\in\mathcal{IR}$. 
    This problem is solved by $(\{g_i^*\}_{i=1}^n,\{c_i^*\}_{i=1}^n)\in\mathcal{G}\times\R^n$ if and only if $\{g_i^*\}_{i=1}^n$ solves
        \begin{equation}
            \label{eq:inf_conv_simplified}
            \sup_{\{g_i\}_{i=1}^n\in\mathcal{G}}\left\{\sum_{i=1}^nU_i(g_i(S-\underline{s}))\right\}\,,
        \end{equation}
    
    \noindent and the constants $c_i^*$ are chosen such that $\{g_i^*(S-\underline{s})+c_i^*\}_{i=1}^n\in\mathcal{IR}$ and $\sum_{i=1}^nc_i^*=\underline{s}$. We will show that the form given in the statement of the theorem is a necessary condition for $\{g_i^*\}_{i=1}^n$ to be a solution to \eqref{eq:inf_conv_simplified}.

    \medskip
    
    By Lemma \ref{lem:coherent_rep}, we may rewrite \eqref{eq:inf_conv_simplified} as follows:
        \begin{align*}
            \sup_{\{g_i\}_{i=1}^n\in\mathcal{G}}\left\{\sum_{i=1}^{n}U_i\left(g_i(S-\underline{s})\right)\right\}
            &=\sup_{\{g_i\}_{i=1}^n\in\mathcal{G}}\left\{\sum_{i=1}^{n}\inf_{T_i\in\mathcal{T}_i}\int g_i(S-\underline{s})\,dT_i\circ\Pr\right\}.
        \end{align*}

\medskip

    For each $i\in\mathcal{N}$, let
        \[
            A_i:=\left\{\int g_i(S-\underline{s})\,dT_i\circ\Pr:T_i\in\mathcal{T}_i\right\}\,.
        \]
    Then we have
        \begin{align}
            \sup_{\{g_i\}_{i=1}^n\in\mathcal{G}}\left\{\sum_{i=1}^{n}\inf_{T_i\in\mathcal{T}_i}\int g_i(S-\underline{s})\,dT_i\circ\Pr\right\}
            &=\sup_{\{g_i\}_{i=1}^n\in\mathcal{G}}\left\{\sum_{i=1}^{n}\inf A_i\right\}\nonumber\\
            &=\sup_{\{g_i\}_{i=1}^n\in\mathcal{G}}\left\{\inf\left(\sum_{i=1}^{n}A_i\right)\right\}\nonumber\\
            &=\sup_{\{g_i\}_{i=1}^n\in\mathcal{G}}\,\inf_{\{T_i\}_{i=1}^n\in\prod_{i=1}^n\mathcal{T}_i}\,\sum_{i=1}^{n}\int g_i(S-\underline{s})\,dT_i\circ\Pr\,,
            \label{eq:minmax_obj}
        \end{align}

\smallskip

\noindent since the infimum commutes with the Minkowski sum $\sum_{i=1}^nA_i$.
    
    \medskip

    Note that the range of $S-\underline{s}$ is contained within the interval $[0,2M]$, where $M$ is the essential supremum norm of $S$. Let $\mathcal{C}([0,2M])$ denote the set of continuous functions on $[0,2M]$, which is a Banach space under the supremum norm. Let $\mathcal{D}:=\R^{[0,1]}$ denote the space of functions from $[0,1]\to\R$, which is a topological vector space with the topology of pointwise convergence. Consider the objective function of \eqref{eq:minmax_obj} as a function from $\mathcal{C}([0,2M])^n\times\mathcal{D}^n$ to $\R$ with the product topology of the spaces $\mathcal{C}([0,2M])$ and $\mathcal{D}$.

    \medskip
    
    Since the Choquet integral is comonotone additive, this objective function is linear in both $\{g_i\}_{i=1}^n$ and $\{T_i\}_{i=1}^n$. Both $\mathcal{G}$ and $\prod_{i=1}^n\mathcal{T}_i$ are convex, with the latter due to Lemma \ref{lem:distortion_set_convex}. Furthermore, since $\mathcal{G}$ is a closed subset of a Cartesian product of 1-Lipschitz functions on the interval $[0,2M]$, it is also compact by the Arzela-Ascoli Theorem \cite[IV.6.7]{Dunford}.

\medskip

    We now verify some continuity properties of the objective function. Firstly, for each $i\in\mathcal{N}$, let $\left\{g_i^{(k)}\right\}_{k=1}^\infty$ be a sequence that converges to $g_i$ uniformly (i.e., with respect to the supremum norm on $\mathcal{C}([0,2M])$). Then $g_i^{(k)}(S-\underline{s})\to g_i(S-\underline{s})$ uniformly on $L^\infty$. The objective function is therefore continuous in $\{g_i\}_{i=1}^n$, since the Choquet integral is continuous with respect to the $L^\infty$ norm.\footnote{In fact, it is Lipschitz continuous \cite[Proposition 4.11]{MarinacciMontrucchio}.} Furthermore, by the proof of Lemma \ref{lem:distortion_set_convex}, the objective is sequentially continuous in each $T_i$ under pointwise convergence. 

\medskip
    
    Therefore, by Sion's minimax theorem (e.g., \cite{komiya1988elementary}), the minimax equality holds for this problem. Exchanging the order of the supremum and infimum yields the minimax problem
        \begin{equation}
            \label{eq:minimax}
            \inf_{\{T_i\}_{i=1}^n\in\prod_{i=1}^n\mathcal{T}_i}\,\max_{\{g_i\}_{i=1}^n\in\mathcal{G}}\,\sum_{i=1}^{n}\int g_i(S-\underline{s})\,dT_i\circ\Pr\,,
        \end{equation}
    where the inner supremum is attained due to compactness. 

    Recall that a solution to \eqref{eq:inf_conv_simplified} exists (see \citealt[Theorem 2.5]{filipovic2008optimal}). Let $\{g_i^*\}_{i=1}^n$ be a solution to \eqref{eq:inf_conv_simplified}. Since the minimax equality holds, it follows from standard results on minimax problems (e.g., \citealt[Section 2.3]{barbuprecupanu}) that for every vector of distortions $\{T_i^*\}_{i=1}^n$ solving \eqref{eq:minimax}, the pair $(\{T_i^*\}_{i=1}^n,\{g_i^*\}_{i=1}^n)$ is a saddle point of \eqref{eq:minimax} as a minimax problem. Hence, it must be true that
        \begin{align*}
            \sum_{i=1}^{n}\int g_i^*(S-\underline{s})\,dT_i^*\circ\Pr
                &=\inf_{\{T_i\}_{i=1}^n\in\prod_{i=1}^n\mathcal{T}_i}\,\max_{\{g_i\}_{i=1}^n\in\mathcal{G}}\,\sum_{i=1}^{n}\int g_i(S-\underline{s})\,dT_i\circ\Pr\\
                &=\max_{\{g_i\}_{i=1}^n\in\mathcal{G}}\,\sum_{i=1}^{n}\int g_i(S-\underline{s})\,dT_i^*\circ\Pr\,.
        \end{align*}
    
    To complete the proof, we will now show that for any fixed vector of convex distortion functions $(T_1,\ldots,T_n)$, we have
        \[
            \max_{\{g_i\}_{i=1}^n\in\mathcal{G}}\,\sum_{i=1}^{n}\int g_i(S-\underline{s})\,dT_i\circ\Pr=\int_0^\infty\max_{i\in\mathcal{N}}\,\{T_i(\Pr(S>\underline{s}+x))\}\,dx\,,
        \]
    and hence problems \eqref{eq:minimax} and \eqref{eq:inf_problem} are equivalent. Furthermore, this maximum is attained at $\{g_i^*\}\in\mathcal{G}$ if and only if $g_i^*(x)=\int_0^xh_i(z)\,dz$, and
         \[
            \sum_{i\in L_x}h_i(x)=1\mbox{ and }\sum_{i\in\mathcal{N}\setminus L_x}h_i(x)=0\,,
        \]
    where
        \begin{align*}
            L_x&:=\left\{i\in\mathcal{N}:
                T_i(\Pr(S>\underline{s}+x))=\max_{j\in\mathcal{N}}\{T_j(\Pr(S>\underline{s}+x))\}\right\}\,.
        \end{align*}
    This is the characterization of optimal allocations for Yaari utilities, and a proof can be found in \cite{liu2020weighted}, for instance. We provide the full argument below for completeness.

    Since each $g_i$ is 1-Lipschitz and non-negative, by a standard result (e.g., \citealt[Lemma 2.1]{zhuang2016marginal}), we may rewrite the Choquet integral as
        \begin{equation}
            \label{eq:inner_sup}
            \max_{\{g_i\}_{i=1}^n\in\mathcal{G}}\,
                \sum_{i=1}^{n}\int_0^\infty T_i(\Pr(S>\underline{s}+x)) \, g_i'(x)\,dx\,.
        \end{equation}
    Here, $g_i'(x)$ is understood as a function for which $g_i(x)=\int_0^xg_i'(z)\,dz$, since each $g_i$ is absolutely continuous. Let $\{g_i^*\}_{i=1}^n\in\mathcal{G}$ satisfy the form given in the statement of the theorem. We first check that $\{g_i^*\}_{i=1}^n\in\mathcal{G}$. By construction, each $g_i^*$ is increasing. Furthermore, for all $x\in\R_+$, we have
        \begin{align*}
            \sum_{i=1}^ng_i^*(x)
            =\sum_{i=1}^n\int_0^xh_i(z)\,dz
            =\int_0^x\sum_{i=1}^nh_i(z)\,dz
            =\int_0^x1\,dz
            =x\,,
        \end{align*}
    and so $\{g_i^*\}_{i=1}^n\in\mathcal{G}$.
    
    Now suppose that $\{\tilde{g}_i\}_{i=1}^n\in\mathcal{G}$. For ease of notation, let $L_x^C:=\mathcal{N}\setminus L_x$. Then we have
    \allowdisplaybreaks
        \begin{align*}
            \sum_{i=1}^nU_i(\tilde{g}_i(S-\underline{s}))
            &=\sum_{i=1}^n\int_0^\infty T_i(\Pr(S>\underline{s}+x))\tilde{g}_i'(x)\,dx
            =\int_0^\infty\sum_{i=1}^nT_i(\Pr(S>\underline{s}+x))\tilde{g}_i'(x)\,dx\\
            &\le\int_0^\infty\sum_{i=1}^n\max_{i\in\mathcal{N}}\,\{T_i(\Pr(S>\underline{s}+x))\}\tilde{g}_i'(x)\,dx\stepcounter{equation}\tag{\theequation}\label{ineq}\\
            &=\int_0^\infty\max_{i\in\mathcal{N}}\,\{T_i(\Pr(S>\underline{s}+x))\}\sum_{i=1}^n\tilde{g}_i'(x)\,dx
            =\int_0^\infty\max_{i\in\mathcal{N}}\,\{T_i(\Pr(S>\underline{s}+x))\}\,dx\\
            &=\int_0^\infty\max_{i\in\mathcal{N}}\,\{T_i(\Pr(S>\underline{s}+x))\}
                +\sum_{i\in L_x^C}\big\{T_i(\Pr(S>\underline{s}+x))\cdot0\big\}\,dx\\
            &=\int_0^\infty\max_{i\in\mathcal{N}}\,\{T_i(\Pr(S>\underline{s}+x))\}\cdot\sum_{i\in L_x}h_i(x)
                +\sum_{i\in L_x^C}\big\{T_i(\Pr(S>\underline{s}+x))\cdot h_i(x)\big\}\,dx\\
            &=\int_0^\infty\sum_{i\in L_x}\big\{\max_{i\in\mathcal{N}}\,\{T_i(\Pr(S>\underline{s}+x))\}\cdot h_i(x)\big\}
                +\sum_{i\in L_x^C}\big\{T_i(\Pr(S>\underline{s}+x))\cdot h_i(x)\big\}\,dx\\
            &=\int_0^\infty\sum_{i\in L_x}\big\{T_i(\Pr(S>\underline{s}+x))\cdot h_i(x)\big\}
                +\sum_{i\in L_x^C}\big\{T_i(\Pr(S>\underline{s}+x))\cdot h_i(x)\big\}\,dx\\
            &=\int_0^\infty\sum_{i=1}^nT_i(\Pr(S>\underline{s}+x))h_i(x)\,dx
            =\sum_{i=1}^n\int_0^\infty T_i(\Pr(S>\underline{s}+x))h_i(x)\,dx\\
            &=\sum_{i=1}^nU_i(g_i^*(S-\underline{s}))\,,
        \end{align*}
    implying that $\{g_i^*\}_{i=1}^n$ solves \eqref{eq:inner_sup}. The above also shows that the optimal value of \eqref{eq:inner_sup} is
        \[
            \int_0^\infty\max_{i\in\mathcal{N}}\,\{T_i(\Pr(S>\underline{s}+x))\}\,dx\,.
        \]

\medskip

    Conversely, suppose that $\tilde{g}_i$ are not of the specified form. That is , we have $\underset{i\in L_x^C}{\sum}\tilde{g}_i'(x)>0$ on a set $\mathcal{A}$ of positive measure. Then for every $x$ in $\mathcal{A}$,
        \begin{align*}
            \sum_{i=1}^nT_i(\Pr(S>\underline{s}+x))\tilde{g}_i'(x)
            &=\sum_{i\in L_x}T_i(\Pr(S>\underline{s}+x))\tilde{g}_i'(x)+\sum_{i\in L_x^C}T_i(\Pr(S>\underline{s}+x))\tilde{g}_i'(x)\\
            &=\max_{i\in\mathcal{N}}\,\{T_i(\Pr(S>\underline{s}+x))\}\tilde{g}_i'(x)+\sum_{i\in L_x^C}T_i(\Pr(S>\underline{s}+x))\tilde{g}_i'(x)\\
            &<\max_{i\in\mathcal{N}}\,\{T_i(\Pr(S>\underline{s}+x))\}\tilde{g}_i'(x)+\sum_{i\in L_x^C}\max_{i\in\mathcal{N}}\,\{T_i(\Pr(S>\underline{s}+x))\}\tilde{g}_i'(x)\\
            &=\max_{i\in\mathcal{N}}\,\{T_i(\Pr(S>\underline{s}+x))\}
            =\sum_{i=1}^n\max_{i\in\mathcal{N}}\,\{T_i(\Pr(S>\underline{s}+x))\}\tilde{g}_i'(x),
        \end{align*}
    where the strict inequality follows because $L_x^C$ is non-empty and $\tilde{g}_i'(x)$ are not all zero for $i\in L_x^C$. Taking the integral over the set $\mathcal{A}$ of positive measure gives
        \[
            \int_\mathcal{A}\sum_{i=1}^nT_i(\Pr(S>\underline{s}+x))\tilde{g}_i'(x)\,dx<\int_\mathcal{A}\sum_{i=1}^n\max_{i\in\mathcal{N}}\,\{T_i(\Pr(S>\underline{s}+x))\}\tilde{g}_i'(x)\,dx\,.
        \]
    Therefore the inequality \eqref{ineq} is strict in this case, implying that $\tilde{g}_i(x)$ does not solve \eqref{eq:inner_sup}. \qed
\end{enumerate}

\bigskip

\begin{remark}
In our application of Sion's minimax theorem in the proof of Theorem \ref{thm:cpo_coherent} above, we show that the objective function \eqref{eq:minmax_obj} is sequentially continuous on each $\mathcal{T}_i$. However, standard statements of Sion's minimax theorem in the literature (e.g., \cite{komiya1988elementary}, \citealt[Theorem 2.132]{barbuprecupanu}) require that the objective function be lower semicontinuous with respect to the topological vector space over which the infimum is taken. Nonetheless, a careful examination of the proof of Sion's minimax theorem shows that sequential lower semicontinuity is sufficient for the result. See, in particular, \cite[Lemma 1]{komiya1988elementary}.

\smallskip

This slight generalization is particularly relevant in our case, since our objective function \eqref{eq:minmax_obj} is not, in general, continuous on each $\mathcal{T}_i$ with respect to the topology of pointwise convergence. This is due to the fact that the dominated convergence theorem does not generalize from sequences of functions to nets of functions. Since the topology of pointwise convergence is not metrizable, it follows that sequential continuity does not necessarily imply continuity.
\end{remark}

\bigskip
\section{A Counterexample to Sufficiency of Theorem \ref{thm:cpo_coherent} (ii)}
\label{sec:not_sufficient_example}

Suppose that there are $n=2$ agents in the market, and that the aggregate endowment $S$ is a continuous random variable. Let $T_1$ and $T_2$ be two convex distortion functions such that $T_1<T_2$ with positive measure, and $T_1>T_2$ also with positive measure. Suppose further that $\int S\,dT_1\circ\Pr\ne\int S\,dT_2\circ\Pr$, and without loss of generality, let $\int S\,dT_1\circ\Pr>\int S\,dT_2\circ\Pr$.

\medskip

Let Agent 1's preferences be represented by a Yaari dual utility with respect to the convex distortion function $T_1$. This preference can be represented by a the singleton set
    \[
        \mathcal{T}_1:=\{T_1\}\,,
    \]
in the sense of Lemma \ref{lem:distortion_set_convex}. That is, for each $Z\in\mathcal{X}$,
    \[
        U_1(Z)=\inf_{T\in\mathcal{T}_1}\int Z\,dT\circ\Pr=\int Z\,dT_1\circ\Pr\,.
    \]
Let $\mathcal{T}_2$ be the closed convex hull of $T_1$ and $T_2$, and suppose that
    \[
        U_2(Z)=\inf_{T\in\mathcal{T}_2}\int Z\,dT\circ\Pr\,.
    \]
Since $T_1\in\mathcal{T}_2$, Agent 2's utility functional is always dominated by Agent 1's utility functional. This implies that it is comonotone Pareto optimal for Agent 1 to assume all of the variability in the market. Indeed, for any $(g_1,g_2)\in\mathcal{G}$, we have
    \[
        U_1(g_1(S))+U_2(g_2(S))\le U_1(g_1(S))+U_1(g_2(S))=U_1(g_1(S)+g_2(S))=U_1(S)\,.
    \]

\medskip

However, applying the result of Theorem \ref{thm:cpo_coherent} does not easily identify the comonotone Pareto-optimal allocations in this scenario. Firstly, since $T_2>T_1$ on a set of positive measure and $S$ is a continuous random variable, it follows that the set 
    \[
        \big\{x\in\R_+:T_2(\Pr(S>\underline{s}+x))>T_1(\Pr(S>\underline{s}+x))\big\}
    \]
also has positive measure. This implies that for each $\lambda\in(0,1]$, we have
\begin{align*}
    \int_0^\infty&\max\big\{T_1(\Pr(S>\underline{s}+x)),\lambda T_2(\Pr(S>\underline{s}+x))+(1-\lambda)T_1(\Pr(S>\underline{s}+x))\big\}\,dx\\
    &=\int_0^\infty T_1(\Pr(S>\underline{s}+x))\,dx+\lambda\int_0^\infty \max\big\{T_2(\Pr(S>\underline{s}+x))-T_1(\Pr(S>\underline{s}+x)),0\big\}\,dx\\
    &>\int_0^\infty T_1(\Pr(S>\underline{s}+x))\,dx\,.
\end{align*}

\noindent Hence, the unique optimizer to \eqref{eq:inf_problem} in this market is the pair $\{T_1,T_1\}$. When applying part (ii) of Theorem \ref{thm:cpo_coherent}, we see that $L_{x,T_1,T_1}=\{1,2\}$ for all $x\in\R_+$. Therefore, the necessary condition for comonotone Pareto optimality provided by Theorem \ref{thm:cpo_coherent} allows for any comonotone allocation in $\mathcal{G}$.

Nonetheless, not all comonotone allocations are Pareto optimal. Consider the two allocations $(S-c^*,c^*)$ and $(\tilde{c},S-\tilde{c})$, where $c^*$ and $\tilde{c}$ are chosen in $\R$ such that both allocations are individually rational. Then we have
    \begin{align*}
        U_1(\tilde{c})+U_2(S-\tilde{c})=U_2(S)=\int S\,dT_2\circ\Pr<\int S\,dT_1\circ\Pr=U_1(S)=U_1(S-c^*)+U_2(c^*)\,,
    \end{align*}
implying that the allocation $(\tilde{c},S-\tilde{c})$ is comonotone but not comonotone Pareto optimal. Hence, the necessary condition provided by Theorem \ref{thm:cpo_coherent} is not sufficient in this case.

\bigskip
\section{Proofs for Section \ref{subsec:equilibrium}}
\label{sec:proofs_equil}

\subsection{Proof of Proposition \ref{prop:equil_exists}}
The result follows from \cite{FilipovicKupper2008} (Theorems 3.1 and 3.2), provided that the preferences are $\sigma(L^\infty,L^1)$ upper semicontinuous and that the value of the sup-convolution problem $\mathcal{S}_\One$ is finite. Since monetary utilities are norm-continuous, law-invariance implies $\sigma(L^\infty,L^1)$ upper semicontinuity by Theorem 2.2 of \cite{JouiniSchachermayerTouzi2006}. Moreover, Theorems \ref{thm:PO_vs_CPO} and \ref{thm:cpo_coherent} imply that the value of $\mathcal{S}_\One$ is the value of \eqref{eq:inf_problem}, which is finite.\qed

\bigskip
\subsection{Proof of Lemma \ref{LemmaEqu}}

This follows from the fact that the inequality in the demand problem \eqref{eq:demand} must hold with equality at optimum. Indeed, if we have $\E^{\Q^*}[Y_i^*+b_i]<\E^{\Q^*}[X_i]$, then
$$U_i\left(Y_i^*+b_i+\E^{\Q^*}[X_i]-\E^{\Q^*}[Y_i^*+b_i]\right)=U_i(Y_i^*+b_i)+\E^{\Q^*}[X_i]-\E^{\Q^*}[Y_i^*+b_i]>U_i(Y_i^*+b_i)\,,$$
and so $Y_i^*+b_i$ does not solve \eqref{eq:demand}. Therefore we must have $\E^{\Q^*}[Y_i^*+b_i]=\E^{\Q^*}[X_i]$, which implies that $b_i=\E^{\Q^*}[X_i-Y_i^*]$.\qed

\bigskip
\subsection{Proof of Proposition \ref{prop:equil_yaari}}
Note that since $(Y_1^*,\ldots,Y_n^*)\in\mathcal{CPO}$, $(Y_1^*+\E^{\Q^*}[X_1-Y_1^*],\ldots,Y_n^*+\E^{\Q^*}[X_n-Y_n^*])$ is an allocation, and so the market-clearing condition is satisfied. Suppose that, for each $i\in\mathcal{N}$, $Y_i^*+\E^{\Q^*}[X_i-Y_i^*]$ solves the demand problem \eqref{eq:demand} under the pricing measure $\Q^*$. Recall from Lemma \ref{lem:translation} that every comonotone allocation $(Y_1,\ldots,Y_n)$ has the form
    \[
        Y_i=g_i(S-\underline{s})+c_i\,,
    \]
where $\{g_i\}_{i=1}^n\in\mathcal{G}$ and $\{c_i\}_{i=1}^n\in\R^n$ satisfies $\sum_{i=1}^nc_i=\underline{s}$. In particular, we can write $Y_i^*+\E^{\Q^*}[X_i-Y_i^*]=g_i^*(S-\underline{s})+c_i^*$, where $\{g_i^*\}_{i=1}^n\in\mathcal{G}$ and $\{c_i^*\}_{i=1}^n\in\R^n$.
Then for each $i\in\mathcal{N}$, since $Y_i^*+\E^{\Q^*}[X_i-Y_i^*]$ solves the demand problem \eqref{eq:demand}, it follows that $g_i^*$ and $c_i^*$ must be a solution to
    \begin{align*}
        \max_{g_i\in\mathcal{I},\,c_i\in\R}&\quad U_i(g_i(S-\underline{s})+c_i)\\
        s.t.&\quad\E^{\Q^*}[g_i(S-\underline{s})+c_i]=\E^{\Q^*}[X_i]\,,
    \end{align*}
where $\mathcal{I}$ is the set of non-decreasing $1$-Lipschitz functions on $\R_+$. By applying translation invariance of $U_i$, we see that $g_i^*$ must be a solution to
    \[
        \max_{g_i\in\mathcal{I}} \ \left\{U_i(g_i(S-\underline{s}))+\E^{\Q^*}[X_i]-\E^{\Q^*}[g_i(S-\underline{s})]\right\}\,.
    \]
Since $\E^{\Q^*}[X_i]$ is a constant, if $g_i^*$ solves the above, it must also be a solution to
    \begin{equation}
        \label{eq:demand_simplified}
        \max_{g_i\in\mathcal{I}} \ \left\{U_i(g_i(S-\underline{s}))-\E^{\Q^*}[g_i(S-\underline{s})]\right\}\,.
    \end{equation}

\smallskip

We now provide an explicit solution to \eqref{eq:demand_simplified}. By rewriting the Choquet integral according to \cite[Lemma 2.1]{zhuang2016marginal}, the objective of \eqref{eq:demand_simplified} becomes
    \begin{align*}
        &\int_0^\infty T_i(\Pr(S>\underline{s}+x))\,g_i'(x)\,dx-
            \int_0^\infty\Q^*(S>\underline{s}+x)\,g_i'(x)\,dx\\
        &=\int_0^\infty\big(T_i(\Pr(S>\underline{s}+x))-
            \Q^*(S>\underline{s}+x)\big)\,g_i'(x)\,dx\,.
    \end{align*}
It follows that $g_i^*$ solves \eqref{eq:demand_simplified} if and only if $g_i^*(x)=\displaystyle\int_0^xh_i(z)\,dz$, where $h_i$ satisfies
    \begin{equation}
        \label{eq:demand_solution}
        h_i(x)=\begin{cases}
            1\,,&\mbox{if }T_i(\Pr(S>\underline{s}+x))>\Q^*(S>\underline{s}+x)\,,\\
            0\,,&\mbox{if }T_i(\Pr(S>\underline{s}+x))<\Q^*(S>\underline{s}+x)\,,
        \end{cases}
    \end{equation}
for almost every $x\in\R_+$.

\smallskip

We conclude the proof by showing that $h_i$ satisfies \eqref{eq:demand_solution} for all $i\in\mathcal{N}$ if and only if
    \begin{equation}
        \label{eq:q_condition}
        T_{(n-1)}(\Pr(S>\underline{s}+x))\le\Q^*(S>\underline{s}+x)\le T_{(n)}(\Pr(S>\underline{s}+x))\,,
    \end{equation}
for almost every $x\in\R_+$. First, suppose that \eqref{eq:q_condition} holds.
Recall from our characterization of $\mathcal{CPO}$ in Corollary \ref{cor:cpo_drm_characterization} that $h_i$ satisfies
    \[
        \sum_{i\in L_x}h_i(x)=1\mbox{ and }\sum_{i\in L_x^C}h_i(x)=0\,.
    \]
If $T_i(\Pr(S>\underline{s}+x))>\Q^*(S>\underline{s}+x)$, then \eqref{eq:q_condition} implies that $L_x=\{i\}$, and hence $h_i(x)=1$. Similarly, if $T_i(\Pr(S>\underline{s}+x))<\Q^*(S>\underline{s}+x)$, then \eqref{eq:q_condition} implies that $i\not\in L_x$, and hence $h_i(x)=0$.

\smallskip

For the converse, suppose that \eqref{eq:q_condition} does not hold. If $\Q^*(S>\underline{s}+x)<T_{(n-1)}(\Pr(S>\underline{s}+s)$ and $h_i$ satisfies \eqref{eq:demand_solution} for all $i\in\mathcal{N}$, then there exists $j,k\in\mathcal{N}$ such that $j\ne k$ and $h_j(x)=h_k(x)=1$. Therefore $\sum_{i=1}^nh_i(x)\ge2\ne1$, contradicting the results of Corollary \ref{cor:cpo_drm_characterization}. Hence, $h_i$ does not satisfy \eqref{eq:demand_solution} for all $i\in\mathcal{N}$. Similarly, if $\Q^*(S>\underline{s}+x)>T_{(n)}(\Pr(S>\underline{s}+s)$, then \eqref{eq:demand_solution} would imply that $\sum_{i=1}^nh_i(x)=0\ne1$, a contradiction. Therefore $h_i$ satisfies \eqref{eq:demand_solution} for all $i\in\mathcal{N}$ if and only if \eqref{eq:q_condition} holds, thereby completing the proof. \qed

\newpage

\bibliographystyle{ecta}
\bibliography{biblio}

\vspace{0.6cm}

\end{document}